\documentclass[10pt]{llncs}

\usepackage{amssymb,amsmath}
\usepackage{graphicx, epsfig}

\def\idtt#1{\ensuremath{\mathtt{#1}}}

\newtheorem{fact}{Fact}
\newtheorem{invariant}{Invariant}

\renewenvironment{proof}{\trivlist\item[]\emph{Proof}:}%
{\unskip\nobreak\hskip 1em plus 1fil\nobreak$\Box$
\parfillskip=0pt%
\endtrivlist}

\newenvironment{itemize*}%
  {\begin{itemize}%
    \setlength{\itemsep}{0pt}%
    \setlength{\parskip}{0pt}%
    \setlength{\parsep}{0pt}%
    \setlength{\topsep}{0pt}%
    \setlength{\partopsep}{0pt}%
  }%
  {\end{itemize}}%

\newcommand{\cL}{{\cal L}}
\newcommand{\cR}{{\cal R}}
\newcommand{\cT}{{\cal T}}

\newcommand{\cD}{{\cal D}}
\newcommand{\cM}{{\cal M}}
\newcommand{\cB}{{\cal B}}

\newcommand{\cV}{{\cal V}}

\newcommand{\rank}{\idtt{rank}}
\newcommand{\sindex}{\idtt{index}}
\newcommand{\leftb}{\idtt{left}}
\newcommand{\rightb}{\idtt{right}}
\newcommand{\replace}{\idtt{replace}}
\newcommand{\locate}{\idtt{locate}}
\newcommand{\rng}{\idtt{rng}}
\newcommand{\new}{\idtt{new}}
\newcommand{\sleft}{\idtt{start}}
\newcommand{\sright}{\idtt{end}}
\newcommand{\Rectify}{\idtt{Rectify}}
\newcommand{\eps}{\varepsilon}

\begin{document}

\title{ A Fast Algorithm for Three-Dimensional Layers of Maxima
 Problem}
\author{Yakov Nekrich\thanks{Department of Computer Science, University of Bonn. Email {\tt yasha@cs.uni-bonn.de}.}
}
\institute{}
\maketitle
\begin{abstract}
  We show that the three-dimensional layers-of-maxima problem can be solved 
in $o(n\log n)$ time in the word RAM model. 
  Our  algorithm runs  in $O(n(\log \log n)^3)$ deterministic time or 
$O(n(\log\log n)^2)$ expected time and uses $O(n)$ space.
  We also describe a deterministic  algorithm that uses optimal $O(n)$ space and solves 
  the three-dimensional layers-of-maxima problem in $O(n\log n)$ time 
  in the pointer machine model. 
\end{abstract} 
\thispagestyle{empty}

\section{Introduction}
A point $p$ \emph{dominates} a point $q$ if each coordinate of $p$ 
is larger than or equals to the corresponding coordinate of $q$. 
A point $p$ is a \emph{maximum point} in a set $S$ if no point of $S$ 
dominates $p$. The maxima set of $S$ is the set of all maximum points 
in $S$. 
In the {\em layers-of-maxima} problem we assign points of a set $S$ 
to layers $S_i$, $i\geq 1$,  according to the dominance relation: 
The first layer of $S$ is defined as the maxima set of $S$, the 
layer $2$ of $S$ is the maxima set of $S\setminus S_1$, 
and the $i$-th layer of 
$S$ is the maxima set of $S\setminus (\cup_{j=1}^{i-1} S_j)$. 
In this paper we show that the three-dimensional layers-of-maxima 
problem can be solved in $o(n\log n)$ time. 

\tolerance=1000
{\bf Previous and Related Work.} 
The algorithm of Kung, Luccio, and Preparata~\cite{KLP75} finds the 
maxima set of a set $S$ in 
$O(n\log n)$ time for $d=2$ or $d=3$ dimensions and 
$O(n\log^{d-2} n)$ time for $d\geq 4$ dimensions.
The  algorithm of Gabow, Bentley, and Tarjan~\cite{GBT84} finds the 
maxima set in  $O(n\log^{d-3}n\log\log n)$ time 
for $d\geq 4$ dimensions. 
Very recently, Chan, Larsen, and P\v{a}tra\c{s}cu~\cite{CLP11} described a
randomized  algorithm that  solves  the $d$-dimensional 
maxima problem (i.e., finds the maxima set) for $d\geq 4$ in 
$O(n\log^{d-3}n)$ time.
Numerous works are devoted to variants of the 
maxima problem in 
different computational models and settings:
In~\cite{BF02}, the authors describe a solution for the three-dimensional 
maxima problem in the cache-oblivious model.
Output-sensitive algorithms and algorithms  that find the maxima 
for a random set of points are described in~\cite{BCL90,Cl94,G94,KS85}.
The two-dimensional problem of maintaining the maxima set under insertions
 and deletions 
is considered in~\cite{K94}; the problem of maintaining the maxima set 
for moving points is considered in~\cite{FGT92}.

The general layers-of-maxima problem appears to be more difficult than 
the problem of finding the maxima set.  
The three-dimensional layers-of-maxima problem can be solved in 
$O(n \log n \log \log n)$ time~\cite{A92} using dynamic fractional 
cascading~\cite{MN90}. The algorithm  
of Buchsbaum and Goodrich~\cite{BG04} runs in $O(n\log n)$ time 
and uses $O(n\log n \log \log n)$ space. 
Giyora and Kaplan~\cite{GK09} described a data structure for point 
location in a dynamic set of horizontal segments and showed how  it 
can be  combined with the approach of~\cite{BG04} 
to solve the three-dimensional layers-of-maxima problem 
in $O(n\log n)$ time and $O(n)$ space. 

The $O(n\log n)$ time is optimal even if we want to find the maxima set in 
two dimensions~\cite{KLP75} provided that we work 
in the infinite-precision computation 
model in which  input values, i.e. point coordinates, can be manipulated 
with algebraic operations and compared. 
On the other hand, 
it is well known that it is possible to achieve 
$o(n\log n)$ time (resp.\ $o(\log n)$ time for searching in a data structure) 
for many one-dimensional as well as for some multi-dimensional 
problems and data structures in other computational models. 
For instance, the grid model, that assumes all coordinates to be integers
in the range $[1,U]$ for a parameter $U$, was extensively studied in
 computational geometry. Examples of problems that can be solved efficiently 
in the grid model are orthogonal range reporting queries~\cite{O88} 
and point location queries 
in a  two- and three-dimensional rectangular subdivisons~\cite{BKS95}.
 In fact, we can use standard techniques 
to show that these  queries can be answered in $o(\log n)$ time when all 
coordinates are arbitrary integers. 
Recently, a number of other important geometric problems was shown to be 
solvable 
in $o(n\log n)$ time (resp.\ in $o(\log n)$ time) in the word RAM model.
An incomplete list\footnote{We note that problems in this list are 
more difficult than the layers-of-maxima problem 
because in our case we process a set of 
axis-parallel segments.} includes Voronoi diagrams and three-dimensional convex 
hulls in $O(n\cdot 2^{O(\sqrt{\log \log n})})$ time~\cite{CP07}, 
two-dimensional point location in $O(\log n/\log \log n)$ 
time~\cite{P06,C06}, and dynamic convex 
hull in $O(\log n/\log \log n)$ time~\cite{DP07}. 
Results for the word RAM model are important because they help us better 
understand 
the structure and relative complexity of different problems and demonstrate 
how geometric information can be analyzed in algorithmically useful ways.

{\bf Our Results.}
In this paper we show that the three-dimensional layers-of-maxima problem 
can be solved in $O(n (\log \log n)^3)$ deterministic 
time and $O(n)$ space in the word RAM 
model. If randomization is allowed, our algorithm runs in 
$O(n(\log\log n)^2)$ expected time. 
For comparison, the fastest 
known deterministic linear space sorting algorithm runs 
in $O(n\log \log n)$ time~\cite{H04}. 
Our result is valid in the word 
RAM computation model, but the time-consuming operations, such as 
multiplications,  are only used during the pre-processing step when we sort 
points by coordinates (see section~\ref{sec:over}). For instance, 
if all points are on the $n\times n\times n$ grid, then our algorithm 
uses exactly the same model as~\cite{O88} or~\cite{BKS95}.

We also describe an algorithm 
that uses $O(n)$ space and solves the three-dimensional layers-of-maxima 
problem in optimal $O(n\log n)$ time in the pointer machine model~\cite{T79}. 
The result of Giyora and Kaplan~\cite{GK09} that achieved the 
same space and time bounds is valid only in the RAM model. 
Thus we present  the first algorithm that solves the three-dimensional
 layers-of-maxima problem in optimal time and space in the pointer 
machine model. 

{\bf Overview.}
Our solution, as well as the previous results,  is based on the sweep
 plane algorithm of~\cite{BG04} described 
in section~\ref{sec:over}. The sweep plane algorithm assigns points 
to layers by answering for each  $p\in S$ a point location query in a 
dynamically maintained staircase subdivision. 
We observe that general data structures for point location in a  
set of horizontal segments cannot be used to obtain an $o(n\log n)$ time
 solution. Even in the word RAM model, no dynamic data structure that
 supports both queries and updates in $o(\log n)$ time is known. 
Moreover, by the lower bound 
of~\cite{AHR98} any data structure for a dynamic set of horizontal segments 
needs $\Omega(\log n/\log \log n)$ time to answer a point location 
query. We achieve a significantly better result using the methods 
described below.

In section~\ref{sec:fast} we describe the data structure for point location 
in a staircase subdivision that supports queries in 
$O((\log \log n)^3)$ time and updates\footnote{
We will describe update operations supported by our data structure in 
sections~\ref{sec:over} and~\ref{sec:fast}.} 
in poly-logarithmic time per segment. This result may be of interest on its 
own.

The data structure of section~\ref{sec:fast} is not sufficient to obtain 
the desired runtime and space usage mainly due to high costs of  update 
operations. To reduce the update time and space usage,  we construct 
auxiliary staircases $\cB_i$, such that: 1. the total number 
of segments in $\cB_i$ and the total number of updates is  $O(n/d)$ 
for a parameter $d=\log^{O(1)} n$; 
2. locating a point $p$ among staircases $\cB_i$ gives us an approximate 
location of $p$ among the original staircases $\cM_i$ (up to $O(d)$ 
staircases).
An efficient method 
for maintaining staircases $\cB_i$, described in section~\ref{sec:stair}, 
is the most technically challenging part of our construction.
In section~\ref{sec:fin} we show how the data structure of
 section~\ref{sec:fast} can be combined with the auxiliary staircases 
approach to obtain an $O(n(\log \log n)^3)$ time algorithm. 
We also sketch how the same approach enables us to obtain 
an $O(n\log n)$ time and $O(n)$ space algorithm in the pointer machine model.


\section{Sweep Plane Algorithm}
\label{sec:over}
Our algorithm is based on the three-dimensional sweep method that is also 
used in ~\cite{BG04}. We move the plane parallel to the $xy$ plane
\footnote{We assume that all points have positive coordinates.} from 
$z=+\infty$ to
 $z=0$ and
 maintain the following invariant: when the $z$-coordinate of the plane 
equals $v$ all points $p$ with $p.z\geq v$ are assigned to their layers of
maxima. 
Here and further $p.x$, $p.y$, and $p.z$ denote the $x$,- $y$-, and
 $z$-coordinates of a point $p$.
Let $S_i(v)$ be the set of points $q$ that belong to the $i$-th layer of
maxima 
such that $q.z>v$; let $P_i(v)$ denote the projection of $S_i(v)$ on the 
sweep plane, $P_i(v)=\{\pi(p)\,|\, p\in S_i(v)\}$ where $\pi(p)$ denotes the 
projection of a point $p$ on the $xy$-plane. For each value of $v$ maximal 
points of  $P_i(v)$  form a staircase $\cM_i$; see Fig.~\ref{fig:sweep}. 
When the $z$-coordinate of the sweep plane is changed from  $v+1$ to $v$, 
we assign all points with $p.z=v$  to their layers of maxima. 
If $\pi(p)$, such that $p.z = v$,  
is dominated by a point from $P_i(v+1)$, then $p$ belongs to the $j$-th 
layer of maxima and $j>i$. If $\pi(p)$, such that $p.z=v$, dominates a point 
on $P_k(v+1)$, then $p$ belongs to the $j$-th layer of maxima and $j\leq k$.
We observe that $\pi(p)$ dominates $P_i(v+1)$ if and only if 
the staircase $\cM_i$ is dominated by  $p$, 
i.e., the vertical ray shot from $p$ in $-y$ direction passes through $\cM_i$.  
Hence, the point $p$ belongs to the layer $i$, such that 
$\pi(p)$ is between the staircase $\cM_{i-1}$ and the staircase $\cM_i$.
This means that we can assign a point to its layer by answering a point 
location query in a staircase subdivision. 
When all $p$ with $p.z=v$ are assigned to their layers, staircases
 are updated.

Thus to solve the layers of maxima problem, we examine points in the 
descending order of their $z$-coordinates.
For each $v$, such that there is at least one $p$ with $p.z=v$, we proceed 
as follows: for every $p$ with $p.z=v$  operation $\locate(p)$ identifies the 
staircase $\cM_i$ 
immediately below $\pi(p)$. If the first staircase below $\pi(p)$ has index 
$i$ ($\pi(p)$ may also lie on $\cM_i$), then 
$p$ is assigned to the $i$-th layer of maxima; if $\pi(p)$ 
is below the lowest staircase $\cM_{j}$, then $p$ is assigned to the  
new layer $j+1$. 
When all points with $p.z=v$ are assigned to their layers, the staircases 
are updated. All points $p$ such that $p.z=v$ are examined 
in the ascending order of their $x$-coordinates. 
If a point $p$ with $p.z=v$ is assigned to layer $i$, 
we perform operation $\replace(p,i)$ that removes all 
points of $\cM_i$ dominated by $p$ and inserts $p$ into $\cM_i$.
If the staircase $i$ does not exist, then instead of $\replace(p,i)$ 
we perform the operation $\new(p,i)$; $\new(p,i)$ 
creates a new staircase $\cM_i$ that consists of one horizontal segment $h$ 
with left endpoint $(0,p.y)$ and right endpoint $\pi(p)$ and one vertical
segment $t$ with upper endpoint $\pi(p)$ and lower endpoint $(0,p.x)$.
See Fig.~\ref{fig:sweep} for an example. 

\begin{figure}[tbh]
  \centering
  \begin{tabular}{ccc}
  \includegraphics[width=.45\textwidth]{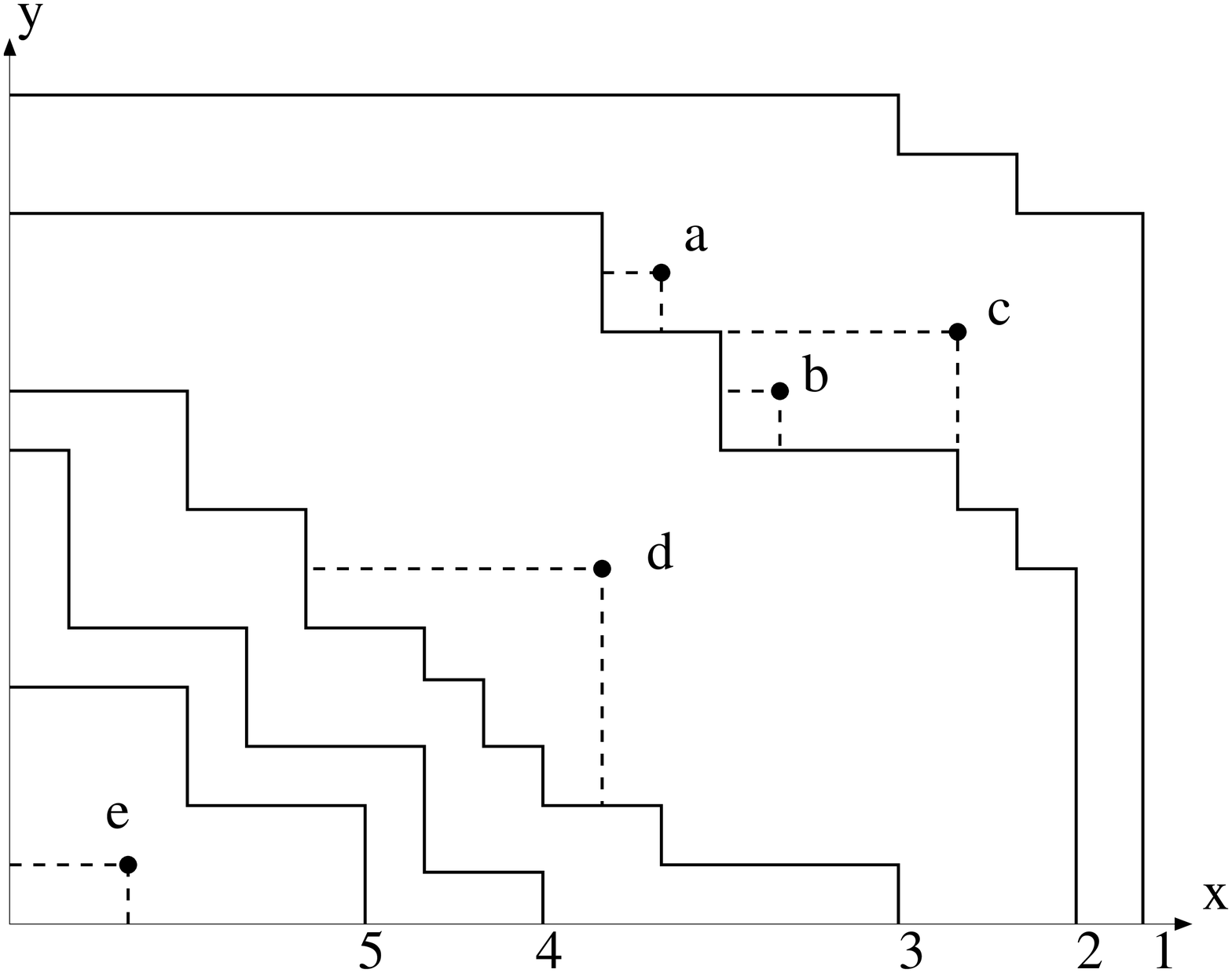} & \hspace*{.7cm} &
  \includegraphics[width=.45\textwidth]{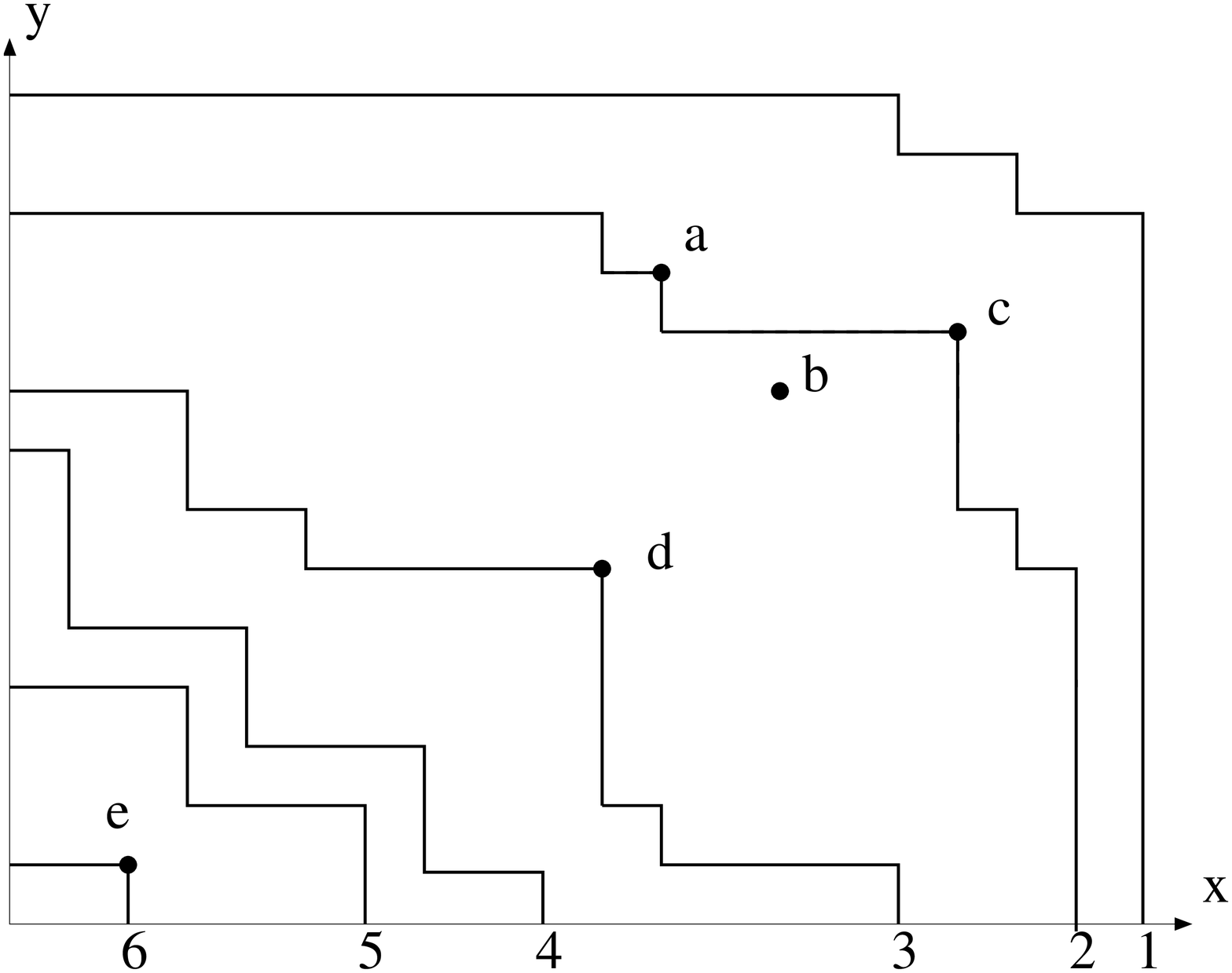} \\
  {\bf (a)}   &  & {\bf(b)}  \\
  \end{tabular}
  \caption{Points $a$, $b$, $c$, $d$, and $e$ have the same $z$-coordinate. {\bf (a)} Points $a$, $b$ and $c$ are assigned to layer $2$, $d$ is assigned to layer $3$, and $e$ is assigned to a new layer $6$. 
{\bf (b)} Staircases after operations $\replace(a,2)$, $\replace(b,2)$, $\replace(c,2)$, $\replace(d,3)$, and $\new(e)$. Observe that $b$ is not the endpoint of a segment in the staircase $\cM_2$ after updates.
     }
  \label{fig:sweep}
\end{figure}

We can reduce the general layers of maxima problem to the problem in the universe of size $O(n)$ using the reduction to rank space
 technique~\cite{O88,GBT84}. 
The rank of an element $e\in S$ is defined as the number of elements in $S$ 
that are smaller than $e$: $\rank(e,S)=|\{a\in S\,|\, a<e \}|$; clearly, 
$\rank(e,S)\leq |S|$. 
For a point $p=(p.x,p.y,p.z)$, $p\in S$, 
 let $\tau(p)=(\rank(p.x,S_x)+1,\rank(p.y,S_y)+1,\rank(p.z,S_z)+1)$. 
Let $S'=\{\tau(p)\,|\, p\in S\}$. Coordinates of all points in $S'$ 
belong to range $[1,n]$. 
A point $p$ dominates a point $q$ if and only if 
$\rank(p.x,S_x)\geq \rank(q.x,S_x)$, $\rank(p.y,S_y)\geq \rank(q.y,S_y)$, and 
$\rank(p.z,S_z)\geq \rank(q.z,S_z)$ where $S_x$, $S_y$, $S_z$ are sets 
of $x$-, $y$-, and $z$-coordinates of points in $S$.  
Hence if a  point $p'\in S'$ is assigned to the 
$i$-th layer of maxima of $S'$, then $\tau^{-1}(p')$ belongs to the $i$-th 
layer of maxima of $S$. 
We can find ranks of $x-$, $y-$, and $z-$coordinates of every point by 
sorting $S_x$, $S_y$, and $S_z$. Using the sorting algorithm of~\cite{H04}, 
$S_x$, $S_y$, and $S_z$ can be sorted in $O(n\log \log n)$ time and $O(n)$ 
space. 
Thus the layers of maxima problem can be reduced to the special case when 
all point coordinates are bounded by $O(n)$ in $O(n\log \log n)$ time.

\section{Fast Queries, Slow Updates}
\label{sec:fast}
In this section we describe a data structure that supports $\locate(q)$ in 
$O( (\log \log n)^3)$ time and update operations $\replace(q,i)$ and 
$\new(q,i)$ in $O(\log n(\log\log n)^2)$ time 
per segment.  We will store horizontal segments of all staircases 
in a data structure that supports \emph{ray shooting queries}: 
given a query point $q$ identify the first segment $s$ crossed by a vertical 
ray that is shot from $q$ in $-y$ direction; in this case 
we will say that the segment $s$ precedes $q$ (or $s$ is the predecessor
 segment of $q$). In the rest of this paper, segments will denote horizontal
 segments.
Identifying the segment that precedes $q$ 
is (almost) equivalent to answering a query $\locate(q)$.
Operation $\replace(q,i)$ corresponds to a deletion of all 
horizontal segments dominated by $q$ and an insertion of at most two 
horizontal segments, see Fig~\ref{fig:sweep}. 
Operation $\new(q,i)$ corresponds to an insertion of a new segment.  

Our data structure is a binary tree on $x$-coordinates and segments are stored 
in one-dimensional  secondary structures in  tree nodes. 
The main idea of our approach is to 
 achieve fast query time by binary search of the root-to-leaf 
path: using properties of staircases, we can determine in 
$O((\log \log n)^2)$ time whether the predecessor segment of a 
point $q$ is stored in the ancestor of a node $v$ or in the descendant 
of a node $v$ for any node $v$ on the path from the root to $q.x$. 
Our approach is similar to the data structure of~\cite{BKS95}, but we need 
additional techniques to support updates.

For a horizontal segment $s$, we denote by $\sleft(s)$ and $\sright(s)$ 
the $x$-coordinates of its left and right endpoints respectively;
we denote by $y(s)$ the $y$-coordinate of all points of $s$.
An integer $e\in S$ precedes (follows) an integer $x$ in $S$ if $e$ is 
the largest (smallest) element in $S$, such that $e\leq x$ ($e\geq x$).
Let $H$ be a set of segments and let $H_y$ be the set of $y$-coordinates 
of segments in $H$. We say that $s\in H$ precedes (follows) an integer 
$e$ if the $y$-coordinate of $s$ precedes (follows) $e$ in $H_y$.
Thus a segment that precedes a point $q$ is a segment that precedes 
$q.y$ in the set of all segments that intersect the vertical line $x=q.x$. 

We construct a balanced binary tree $\cT$ of height $\log n$ on the set of 
all possible $x$-coordinates, i.e., $n$ leaves of $\cT$ correspond to  
integers in $[1,n]$. The  range of a node $v$ is the interval 
$\rng(v)=[\leftb(v),\rightb(v)]$ where $\leftb(v)$ and $\rightb(v)$ are 
leftmost and rightmost leaf descendants of $v$.  

We say that a segment $s$ spans a node $v$ if 
$\sleft(s) < \leftb(v) < \rightb(v) < \sright(s)$; 
a segment $r$ belongs to a node $v$ if 
$\leftb(v) < \sleft(s) < \sright(s) < \rightb(v)$.
A segment $s$ $l$-cuts a node 
$v$ if $s$ intersects the  vertical line $x=\leftb(v)$, 
but $s$ does not span $v$, i.e.,
$\sleft(s)\leq \leftb(v)$ and $\sright(s) < \rightb(v)$; 
a segment $s$ $r$-cuts a node $v$ if
$s$ intersects the vertical line  $x=\rightb(v)$ but $s$ does not span $v$, 
i.e., $\sleft(s) >  \leftb(v)$ and $\sright(s) \geq \rightb(v)$. 
A segment $s$ such that $[\sleft(s),\sright(s)]\cap \rng(v)\not=\emptyset$ 
either cuts $v$, or spans $v$, or belongs to $v$. 
We store $y$-coordinates of all  segments that 
$l$-cut ($r$-cut) a node $v$  in a data structure $\cL_v$ 
($\cR_v$). Using exponential trees~\cite{AT07}, we can  implement 
$\cL_v$ and $\cR_v$ in linear space, so that one-dimensional searching 
(i.e. predecessor and successor queries) is supported in $O((\log \log n)^2)$
time. 
Since a segment cuts $O(\log n)$ nodes (at most two  nodes on each tree 
level), all $\cL_v$ and $\cR_v$ use $O(n\log n)$ space. 
We denote by $\sindex(s)$ the index of the staircase $\cM_i$ that contains 
$s$, i.e., $s\in \cM_{\sindex(s)}$.
The following simple properties are important for the search procedure:
\begin{fact}\label{fact:vert} 
Suppose that an arbitrary vertical line cuts staircases $\cM_i$ and $\cM_j$, 
$i<j$, in points $p$ and $q$ respectively. Then $p.y > q.y$ because 
staircases do not cross.  
\end{fact}
\begin{fact}\label{fact:monot}
For any two points $p$ and $q$ on a staircase $\cM_i$,
if $p.x < q.x$, then $p.y \geq q.y$
\end{fact}
\begin{fact}\label{fact:below}
Given a staircase $\cM_i$ 
and a point $p$, we can determine whether  $\cM_i$ 
is below or above $p$ and find the segment $s\in \cM_i$ 
such that $p.x\in [\sleft(s),\sright(s)]$ in $O((\log \log n)^2)$ time.
The data structure $D_i$ that supports such queries 
uses linear space and supports finger updates in $O(1)$ time.
\end{fact}
\begin{proof}
The data structure $D_i$ contains $x$-coordinates 
of all segment endpoints of $\cM_i$. 
$D_i$ is implemented as an exponential tree so that 
it uses $O(n)$ space. 
Using $D_i$ we can identify $s\in \cM_i$ such that 
$p.x \in [\sleft(s),\sright(s)]$ in $O((\log \log n)^2)$ time;
$\cM_i$ is below $p$ if and only if $s$ is below $p$. 
\end{proof}
Using Fact~\ref{fact:below}
 we can determine whether a segment $s$ precedes a point $q$ 
in $O((\log\log n)^2)$ time: Suppose that $s$ belongs to a staircase 
$\cM_i$. Then $s$ is the predecessor segment of $q$ iff $q.x\in [\sleft(s),\sright(s)]$, $q.y\geq y(s)$ and the 
staircase $\cM_{i-1}$ is above $q$. \\
We can use these properties and data structures
$\cL_v$ and $\cR_v$ to determine whether a segment $b$ that precedes
 a point $q$ spans 
a node $v$, belongs to a node $v$, or cuts a node $v$. If the segment 
$b$ we are looking for spans $v$, then it cuts an ancestor of $v$; if that 
segment belongs to $v$, then it cuts a descendant of $v$. Hence, we can 
apply binary search and find in $O(\log \log n)$ iterations the node 
$f$ such that the predecessor segment of $q$ cuts $f$. 
Observe that in some situations there may be no staircase $\cM_i$ below 
$q$, see  Fig~\ref{fig:nobelow} for an example. 
To deal with such situations, we insert 
a dummy segment $s_d$ with left endpoint $(1,0)$ and right endpoint
$(n,0)$; we set $\sindex(s_d)=+\infty$ and store $s_d$ in the data structure 
$\cL_{v_0}$ where $v_0$ is the root of $\cT$.

Let $l_x$ be the leaf in which the predecessor of $q.x$ is stored. 
We will use variables $l$, $u$ and $v$ to guide the search for the node $f$. 
Initially we set $l=l_x$ and  $u$ is  the root of $\cT$. 
We set $v$ to be  the 
middle node between $u$ and $l$: if the path between $u$ and $l$ consists of 
$h$ edges, then the path from $u$ to $v$ consists of $\lfloor h/2\rfloor$ 
edges and $v$ is an ancestor of $l$.

Let $r$ and $s$ denote the segments in $\cL_v$ that precede and 
follow $q.y$.
If there is no segment $s$ in $\cL_v$ with $y(s)>q.y$, then 
we set $s=NULL$. 
If there is no segment $r$ in $\cL_v$ with $y(r)\leq q.y$, then 
we set $r=NULL$. 
We can find both $r$ and $s$ in $O((\log \log n)^2)$ time. 
If the segment $r\not= NULL$, we check whether the staircase  
$\cM_{\sindex(r)}$ contains the predecessor segment of $q$; by Fact~\ref{fact:below}, 
this can be done in $O((\log\log n)^2)$ time. 
If $\cM_{\sindex(r)}$ contains the predecessor segment of $q$, 
the search is completed. 
Otherwise, the staircase $\cM_{\sindex(r)-1}$ is below $q$ or 
$r=NULL$. 
In this case we find the segment $r'$ that precedes $q.y$ in $\cR_v$. 
If $r'$ is not the predecessor segment of $q$ or $r'=NULL$, 
then the predecessor segment of $q$ either spans $v$ or belongs to $v$. 
We distinguish between the following two cases:\\ 
{\bf 1.} The segment $s\not=NULL$  and the staircase that contains $s$ is
 below $q$. By Fact~\ref{fact:vert}, a vertical line 
$x=q.x$ will cross the staircase of $s$ before it will cross 
a staircase $\cM_i$, $i> \sindex(s)$. Hence, a segment that spans $v$ and 
belongs to the staircase $\cM_i$, $i> \sindex(s)$, cannot 
be the predecessor segment of $q$. If a segment $t$ spans $v$ and 
$\sindex(t) < \sindex(s)$, then the $y$-coordinate of $t$ is larger than the 
$y$-coordinate of $s$ by Fact~\ref{fact:vert}. Since $y(t)> y(s)$ and 
$y(s)> q.y$,
the segment $t$ is above $q$. 
Thus no segment that spans $v$ can be the predecessor 
of $q$. \\
{\bf 2.} The staircase that contains $s$ is above $q$ or $s=NULL$. 
If $r$ exists, the staircase 
 $\cM_{\sindex(r)-1}$ is below $q$. Hence, the predecessor segment of $q$ 
belongs  to a staircase\footnote{To simplify the description,
 we assume that $\sindex(s)=0$ if $s=NULL$.} 
$\cM_i$, $\sindex(s) < i \leq \sindex(r)-1$.
Since each staircase $\cM_i$, $\sindex(s) < i \leq \sindex(r)-1$, 
 contains a segment that spans $v$, the predecessor segment of $s$ is 
a segment that spans $v$. 
If $r$ does not exist, then every segment below the point $q$ spans 
the node $v$. Hence, the predecessor segment of $s$ spans $v$.
See Fig.~\ref{fig:pl} for an example.

\begin{figure}[tb]
  \centering
  \begin{tabular}{ccccc}
  \includegraphics[width=.25\textwidth]{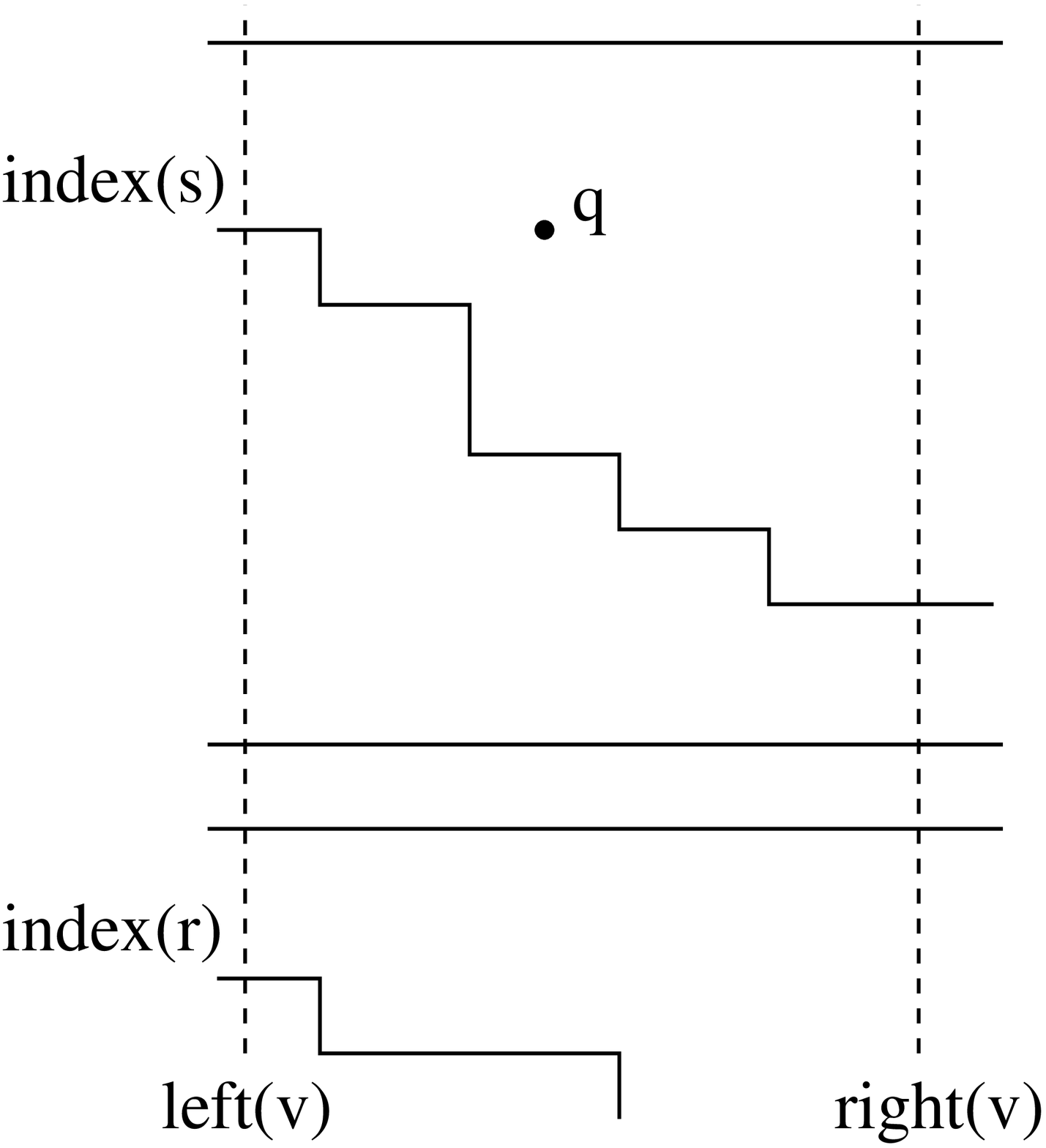} & \hspace*{.2cm} &
  \includegraphics[width=.25\textwidth]{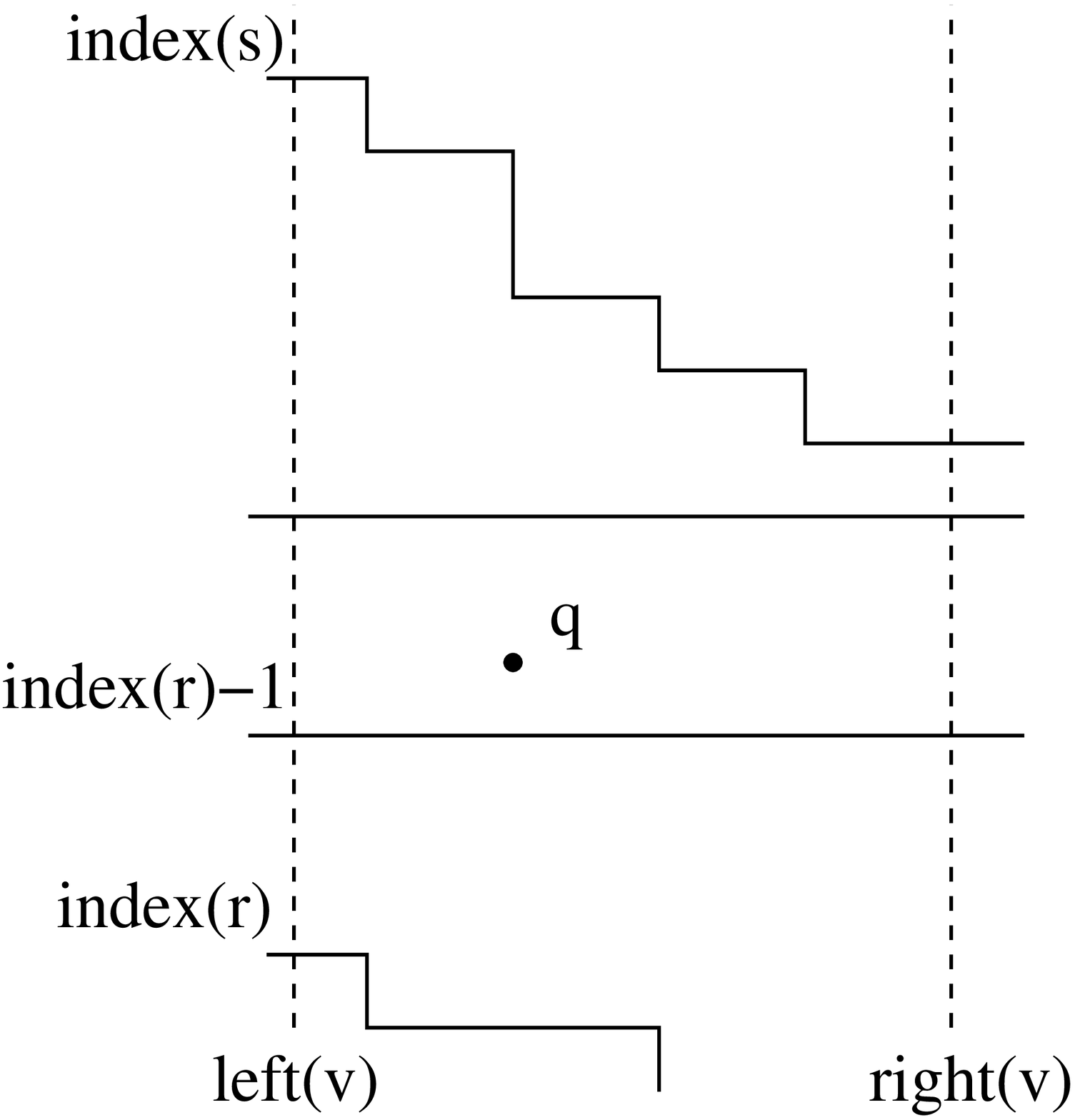} & \hspace*{.2cm} 
  \includegraphics[width=.25\textwidth]{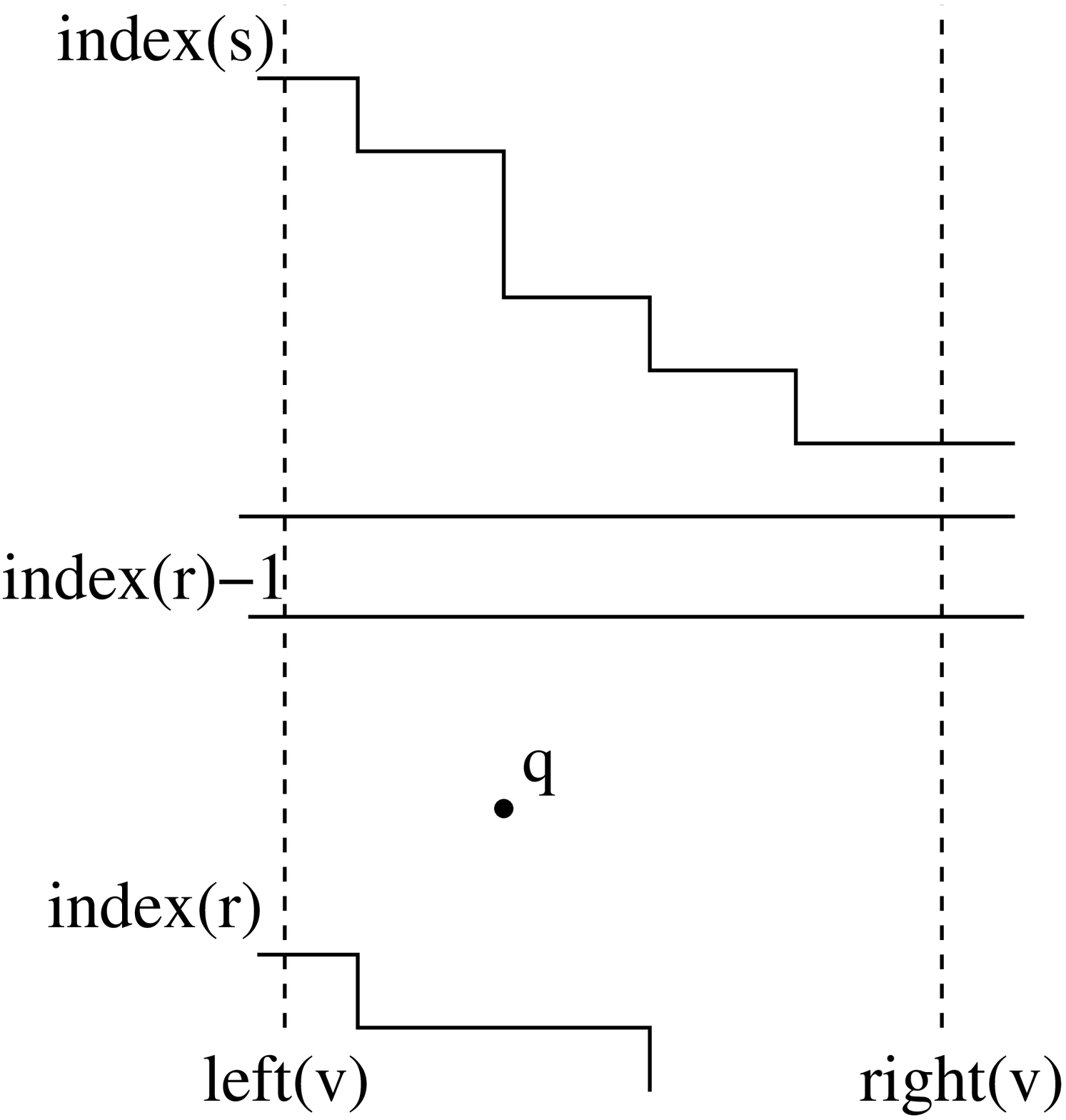}  \\
  {\bf (a)}   &  & {\bf(b)} &  {\bf (c)} \\
  \end{tabular}
  \caption{Search procedure in a node $v$. 
    Staircases are denoted by their indexes. 
    Figures {\bf (a)} and  {\bf (b)} correspond to cases {\bf 1} and {\bf 2} 
    respectively. The case when the predecessor segment belongs to 
    $\cM_{\sindex(r)}$ is shown on Fig. {\bf (c)}.
     }
  \label{fig:pl}
\end{figure}

If the  predecessor segment spans $v$, we search for $f$ among  ancestors 
of $v$;
if the  predecessor segment belongs to $v$, we search for $f$ among 
descendants 
of $v$.  Hence, we set $l=v$ in case 2, and we set $u=v$ in 
case 1. 
Then, we set $v$ to be the middle node between $u$ and $l$ and examine 
the new node $v$. Since we examine $O(\log \log n)$ nodes and spend 
$O((\log \log n)^2)$ time in each node, the total query time is 
$O((\log \log n)^3)$.

If the predecessor segment is the dummy segment $s_d$, then 
there is no horizontal segment of any $\cM_i$ below $q$. 
In this case we must identify the staircase to the left of $q.x$.
Let $m_i$ denote the rightmost point on the staircase $\cM_i$, 
i.e., $m_i$ is a point on $\cM_i$ such that  $m_i.y=0$. 
Then $q$ is between staircases $\cM_{i-1}$ and $\cM_i$, such that 
$m_i.x < q.x < m_{i-1}.x$. We can find $m_i$ in $O((\log\log n)^2)$ 
time. 

When a segment $s$ is deleted, we delete it from the corresponding data
 structure $D_i$. We also delete $s$ from all data structures $\cL_v$ 
and $\cR_w$ for all nodes $v$ and $w$, such that $s$ $l$-cuts $v$  
(respectively $r$-cuts $w$). Since a segment cuts $O(\log n)$ nodes 
and exponential trees support updates in $O((\log\log n)^2)$ time, a 
deletion takes $O(\log n(\log\log n)^2)$ time. 
Insertions are supported in the same way\footnote{The update time can be 
slightly improved using fractional cascading and similar techniques, but 
this is not necessary for our presentation.}. 
Operation $\new(q,l)$ is implemented by inserting a segment with 
endpoints $(0,q.y)$ and $(q.x,q.y)$ into $\cT$, incrementing by one the 
number of staircases $l$, and creating the new data structure $D_l$. 
To implement $\replace(q,i)$ we delete the segments ``covered'' 
by $q$ from $\cT$ and $D_i$ and insert the new segment (or two new segments)
 into $\cT$ and $D_i$.
\begin{lemma}\label{lemma:slow}
We can store $n$ horizontal staircase segments with endpoints 
on $n\times n$ grid in a $O(n\log n)$ space data
 structure that answers ray shooting queries in $O((\log \log n)^3)$ time 
and supports operation $\replace(q,i)$  in $O(m\log n(\log\log n)^2)$ time 
where $m$ is the number of segments inserted into and deleted from the 
staircase $\cM_i$, and operation $\new(q)$ in $O(\log n(\log\log n)^2)$ time.  
\end{lemma}
The data structure of Lemma~\ref{lemma:slow} is deterministic. We can further
improve the query time if randomization is allowed. 
\begin{fact}\label{fact:randbelow}
Given a staircase $\cM_i$ 
and a point $p$, we can determine whether  $\cM_i$ 
is below or above $p$ and find the segment $s\in \cM_i$ 
such that $p.x\in [\sleft(s),\sright(s)]$ in $O((\log \log n))$ time.
The data structure $D_i$ that supports such queries 
uses linear space and supports finger updates in $O(1)$ expected time.
\end{fact}
\begin{proof}
The data structure is the same  as in the proof of 
Fact~\ref{fact:below}, but we use the y-fast tree data structure~\cite{W83} 
instead of the exponential tree.
\end{proof}
\begin{lemma}\label{lemma:slowrand}
We can store $n$ horizontal staircase segments with endpoints 
on $n\times n$ grid in a $O(n\log n)$ space data
 structure that answers ray shooting queries in $O((\log \log n)^2)$ time 
and supports operation $\replace(q,i)$  in $O(m\log n \log\log n)$ expected 
time 
where $m$ is the number of segments inserted into and deleted from the 
staircase $\cM_i$, and operation $\new(q)$ in $O(\log n \log\log n)$ expected 
time.  
\end{lemma}
\begin{proof}
Our data structure is the same as in the proof of Lemma~\ref{lemma:slow}. 
But we implement $D_i$ using Fact~\ref{fact:randbelow}. Data structures 
$\cL_v$ and $\cR_v$ are implemented using  the y-fast tree~\cite{W83}. 
Hence, the search procedure spends $O(\log \log n)$ time in each node 
of $\cT$ and a query is answered in $O((\log\log n)^2)$ time. 
\end{proof}

Although this is not necessary for further presentation, 
we can prove a similar result for the case when all segment endpoints are 
on a $U\times U$ grid; the query time is $O(\log\log U + (\log \log n)^3)$ 
and the update time is $O(\log^3 n(\log\log n)^2)$ per segment. 
See Appendix D for a proof of this result.

\section{Additional Staircases}
\label{sec:stair}
\tolerance=1500
The algorithm in the previous section needs $O(n\log n(\log\log n)^2)$ time to construct
 the layers of maxima: $n$ ray shooting queries can be performed in 
$O(n(\log \log n)^3)$ time, but $O(n)$ update operations take 
$O(n \log n(\log\log n)^2)$ time. To speed-up the algorithm and improve the 
space usage, we  reduce the number 
of updates and the number of segments in the data structure of Lemma~\ref{lemma:slow} to $O(n/\log^2 n)$. 

Let $\cD$ denote the data structure of Lemma~\ref{lemma:slow}. We 
construct and maintain a new sequence of staircases
 $\cB_1,\cB_2,\ldots, \cB_m$, where  $m\leq n/d$ and the parameter $d$ will be
 specified later. All horizontal segments of $\cB_1,\ldots,\cB_m$ are stored 
in $\cD$. The new staircases satisfy the following conditions: \\
1. There are $O(\frac{n}{d})$ horizontal segments in all staircases $\cB_i$\\
2. $\cD$ is updated $O(\frac{n}{d})$ times during the execution of the 
sweep plane algorithm.\\
3. For any point $q$ and for any $i$, if $q$ is between $\cB_{i-1}$ and 
$\cB_{i}$, then $q$ is situated between $\cM_k$ and $\cM_{k+1}$ 
for $(i-3/2)d \leq k \leq (i+1/2)d$. \\
Conditions 1 and 2 imply that the data structure $\cD$ uses $O(n)$ space 
and all updates of $\cD$ take $O(n)$ time if $d\geq \log n(\log\log n)^2$. 
Condition 3 means that we can 
use staircases $\cB_i$ to guide the search among $\cM_k$: we first identify 
the index $i$, such that the query point $q$ is between $\cB_{i+1}$ and 
$\cB_i$, and then locate $q$ in  $\cM_{(i-3/2)d},\ldots, \cM_{(i+1/2)d}$. 
It is not difficult to construct $\cB_i$ that satisfy conditions 1 and 3. 
The challenging part is maintaining the staircases $\cB_i$ with a small number 
of updates.  
\begin{lemma}\label{lemma:stair}
The total number of inserted and deleted segments in all $\cB_i$ is 
$O(\frac{n}{d})$. The number of segments stored in $\cB_i$ is 
$O(\frac{n}{d})$.
\end{lemma}
We describe how staircases can be  maintained and prove Lemma~\ref{lemma:stair} in Appendix B.
\section{Efficient Algorithms for the Layers-of-Maxima Problem}
\label{sec:fin}
{\bf Word RAM Model.}
To conclude the description of our main algorithm, we need the following 
simple 
\begin{lemma}\label{lemma:small}
Using a $O(m)$ space data structure, we can locate a point 
in a group of $d$ staircases $\cM_j,\cM_{j+1},\ldots,
\cM_{j+d}$ in $O(\log d\cdot (\log \log m)^2)$ time, where $m$ is the 
number of segments in $\cM_j,\cM_{j+1},\ldots, \cM_{j+d}$.
An operation $\replace(q,i)$ is supported in $O((\log \log m)^2+m_q)$ 
time, where $m_q$ is the number of inserted and deleted segments in the 
staircase $\cM_i$, $j\leq i \leq i+d$.
\end{lemma}
\begin{proof}
We can use Fact~\ref{fact:below} to determine whether a staircase 
is above or below a staircase $\cM_k$ for any $j\leq k \leq j+d$. 
Hence, we can locate a point in $O(\log d\cdot (\log \log m)^2)$ 
time by a binary search among $d$ staircases.
\end{proof}
We set $d=\log^2 n$. 
The data structure $F_i$ contains all segments of staircases 
$\cM_{(i-1)d+1},\cM_{(i-1)d+2},\ldots,\cM_{id}$ for 
$i=1,2,\ldots,j$, where $j=\lfloor l/d\rfloor$ and $l$ is the highest 
index of a staircase; the data structure $F_{j+1}$ 
contains all segments of staircases $\cM_{jd+1},\ldots, \cM_l$. 
We can locate a point $q$ in each $F_i$ in $O((\log \log n)^3)$ time 
by Lemma~\ref{lemma:small}. 
Since each staircase belongs to one data structure, all 
$F_i$ use $O(n)$ space. 
We also maintain additional staircases $\cB_i$ as described in 
section~\ref{sec:stair}.
All segments of all staircases $\cB_i$ are stored in the 
data structure $\cD$ of Lemma~\ref{lemma:slow}; since $\cD$ contains 
$O(n/d)$ segments, the space usage of $\cD$ is $O(n)$.  

Now we can describe how operations $\locate$, $\replace$, $\new$ can 
be implemented in $O((\log \log n)^3)$ time per segment. 
\begin{itemize*}
\item
$\locate(q)$: We find the index $k$, such that $q$ is between  
$\cB_{k-1}$ and $\cB_k$ in $O((\log \log n)^3)$ time. 
As described in section~\ref{sec:stair}, $q$ is between $\cM_{kd+g}$ 
and $\cM_{(k-1)d-g}$. Hence, we can use  data structures
$F_{k+1}$, $F_k$, and $F_{k-1}$ to identify $j$ such that 
$q$ is between  $\cM_j$ and $\cM_{j+1}$. 
Searching $F_{k+1}$, $F_k$, and $F_{k-1}$ takes $O((\log \log n)^3)$ 
time, and the total time for $\locate(q)$ is $O((\log \log n)^3)$.
\item
$\replace(q, i)$: let $m_q$ be the number of inserted and deleted 
segments. The data structure $F_{\lfloor i/d \rfloor}$ 
can be updated in $O(m_q+ (\log\log n)^2)$ time. We may also have to update 
$\cB_{\lfloor i/d \rfloor}$, $\cB_{\lfloor i/d \rfloor+1}$, and the data structure
 $\cD$.
\item
$\new(q,l)$: If $l=kd+1$ for some $k$, a new data structure $F_{k+1}$ is
 created. We add the horizontal segment of the new staircase into 
the data structure $F_{k+1}$. 
If $l=kd$, we create a new staircase $\cB_k$ and add the segments of 
$\cB_k$ into the data structure $\cD$.
\end{itemize*}
There are $O(n/d)$ update operations on the data structure $\cD$
that can be performed in $O((n/d)\log n(\log\log n)^2)=O(n)$ time. 
If we ignore the time to update $\cD$, then $\replace(q,i)$ takes 
$O(m_q(\log\log n)^2)$ time and $\new(q,l)$ takes $O((\log\log n)^2)$ time. 
Since $\sum_{q\in S} m_q=O(n)$ and $\new(q,l)$ is performed 
at most $n$ times, the algorithm runs in $O(n (\log \log n)^3)$ 
time. We thus obtain the main result of this paper.
\begin{theorem}\label{theor:determ}
The three-dimensional layers-of-maxima problem can be solved in 
$O(n (\log \log n)^3 )$ deterministic time in the word RAM model. 
The space usage of the algorithm is $O(n)$.
\end{theorem}
If we use Fact~\ref{fact:randbelow} instead of Fact~\ref{fact:below} 
in the proof of Lemma~\ref{lemma:small} and Lemma~\ref{lemma:slowrand} 
instead of Lemma~\ref{lemma:slow} in the proof of Theorem~\ref{theor:determ}, 
we  obtain a slightly better randomized algorithm. 
\begin{theorem}
The three-dimensional layers-of-maxima problem 
can be solved in $O(n (\log \log n)^2 )$ expected time. 
The space usage of the algorithm is $O(n)$.
\end{theorem}

{\bf Pointer Machine Model.}
We can apply the idea of additional staircases to obtain  
an $O(n\log n)$ algorithm in the pointer machine model. This time, we set 
$d=\log n $ and maintain additional staircases $\cB_i$ as described 
in section~\ref{sec:stair}.
Horizontal  segments of all $\cB_i$ are stored in the data 
structure $\cD$ of Giyora and Kaplan~\cite{GK09} that uses 
$O(m\log^{\eps} m)$ space and
supports queries and updates in $O(\log m)$ and $O(\log^{1+\eps} m)$ 
time respectively, where $m$ is the number of segments 
in all $\cB_i$ and $\eps$ is an arbitrarily small  positive constant. 
Using dynamic fractional cascading~\cite{MN90}, we can implement $F_i$ so that 
$F_i$ uses linear space and answers queries 
in $O(\log n + \log \log n \log d)= O(\log n)$ time.
Updates are supported in $O(\log n)$ time; 
details will be given in the 
full version of this paper. 
Using $\cD$ and $F_i$, we can implement the 
sweep plane algorithm in the same way as described in the 
first part of this section. The space usage of all data structures 
$F_i$ is $O(n)$, and all updates of $F_i$ take $O(n\log n)$ time. 
By Lemma~\ref{lemma:stair}, the data structure $\cD$ is updated 
$O(n/\log n)$ times; hence all updates of $\cD$ take $O(n\log n)$ time. 
The space usage of $\cD$ is $O(m\log^{\eps} m)= O(n)$. 
Each new point is located by answering one query to $\cD$ and at 
most three queries to $F_i$; hence, a new point is assigned 
to its layer of maxima in $O(\log n)$ time. 
\begin{theorem}
A three-dimensional layers-of-maxima problem can be solved in 
$O(n  \log n)$ time in the pointer machine model. The space usage of
 the algorithm is $O(n)$.
\end{theorem}

\section*{Acknowledgment}
The author wishes to thank an anonymous reviewer of this paper for a 
stimulating comment that helped to obtain the randomized version 
of the presented algorithm.

\newpage

\section*{Appendix A. Figures}

\begin{figure}[htb]
  \centering
  \includegraphics[width=.6\textwidth]{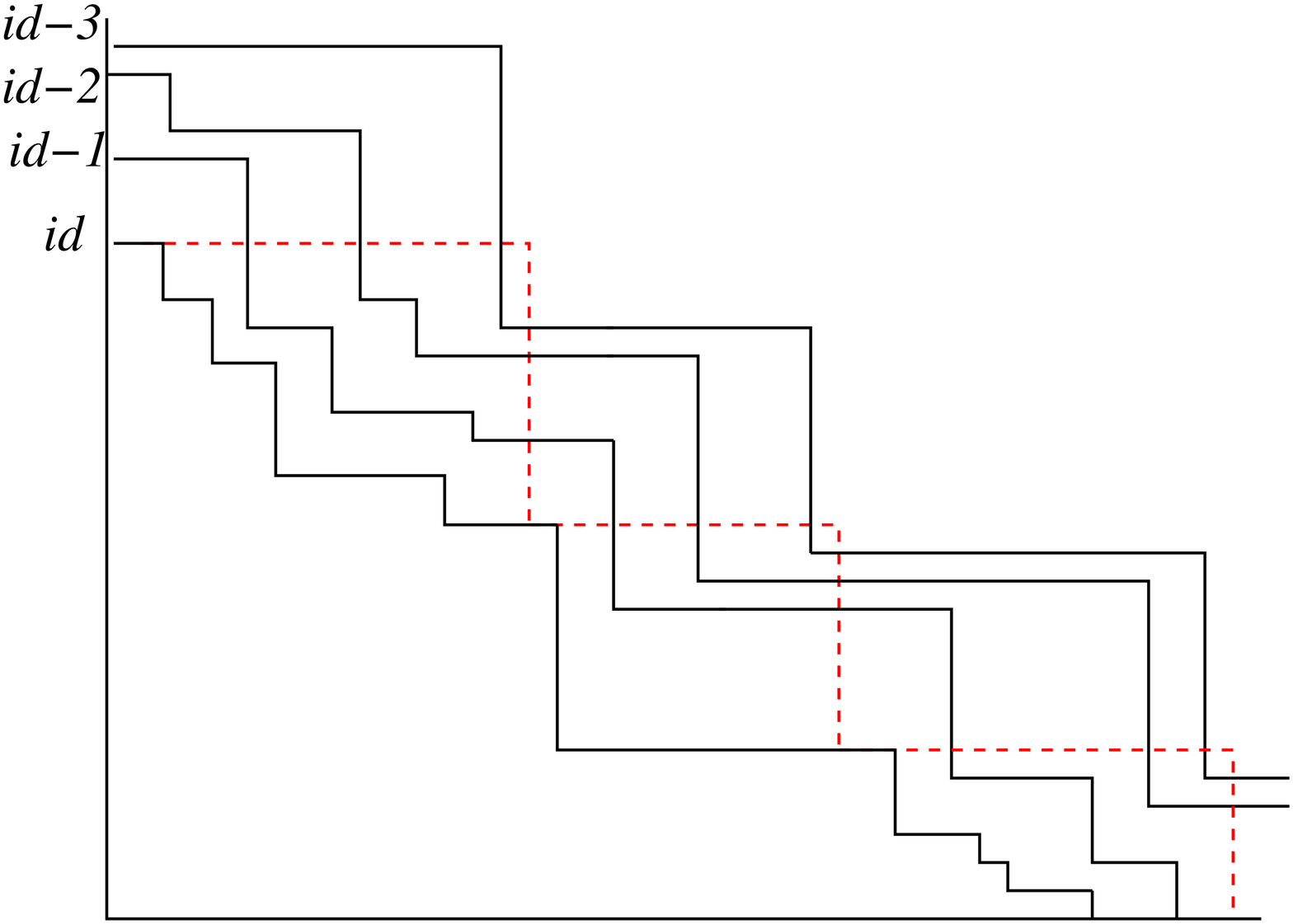}
\caption{
  Example of a just constructed additional staircase $\cB_i$ for $d=6$. 
  The staircase $\cB_i$ is shown with dashed red lines. 
  Staircases are denoted by their indexes. 
 }
  \label{fig:stair1}
\end{figure} 

\begin{figure}[bt]
  \centering
\includegraphics[width=.4\textwidth]{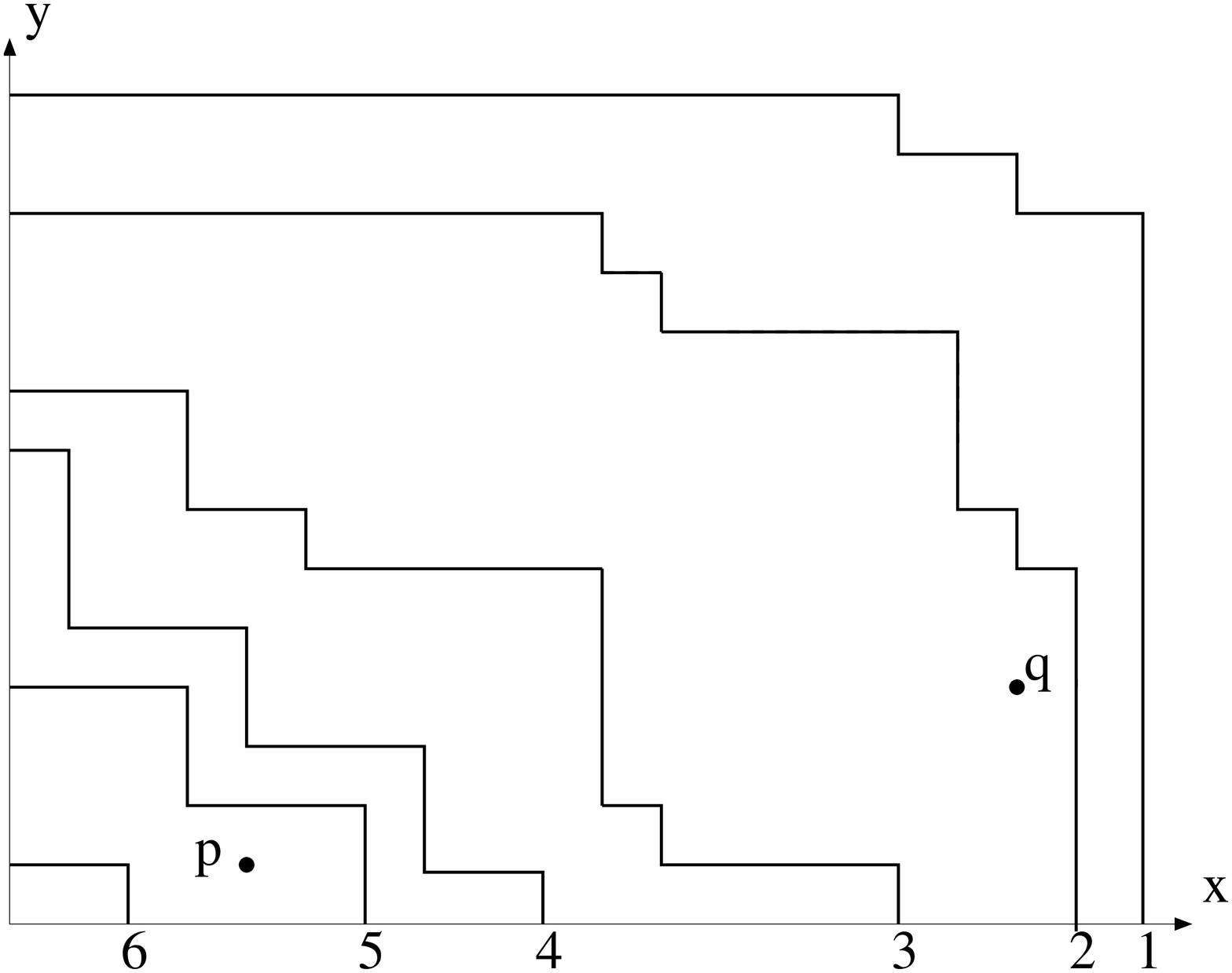}
\caption{There are no staircases below $p$ and $q$.}
\label{fig:nobelow}
\end{figure}

\begin{figure}[tbh]
  \centering
  \begin{tabular}{ccc}
  \includegraphics[width=.45\textwidth]{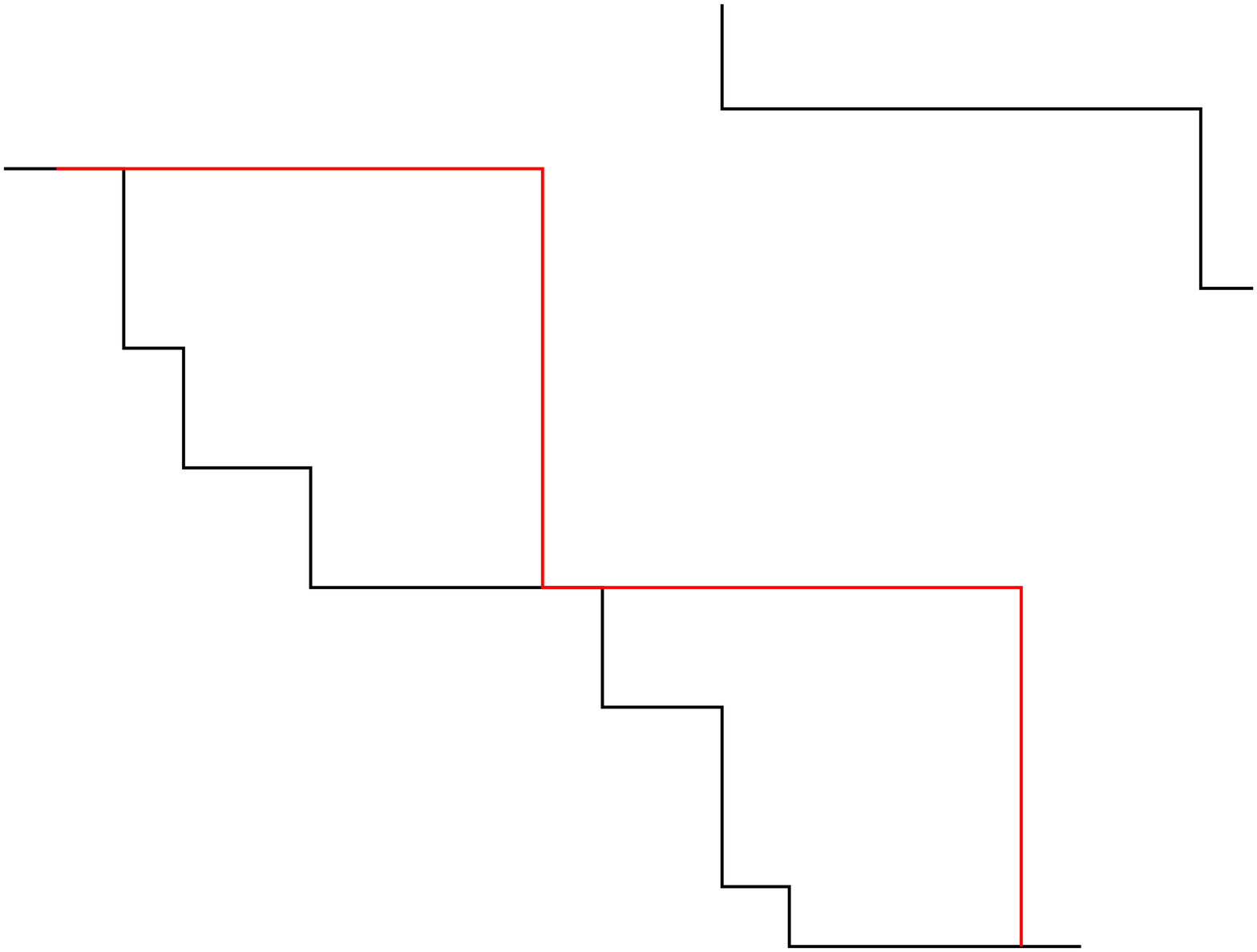} & \hspace*{.7cm} &
  \includegraphics[width=.45\textwidth]{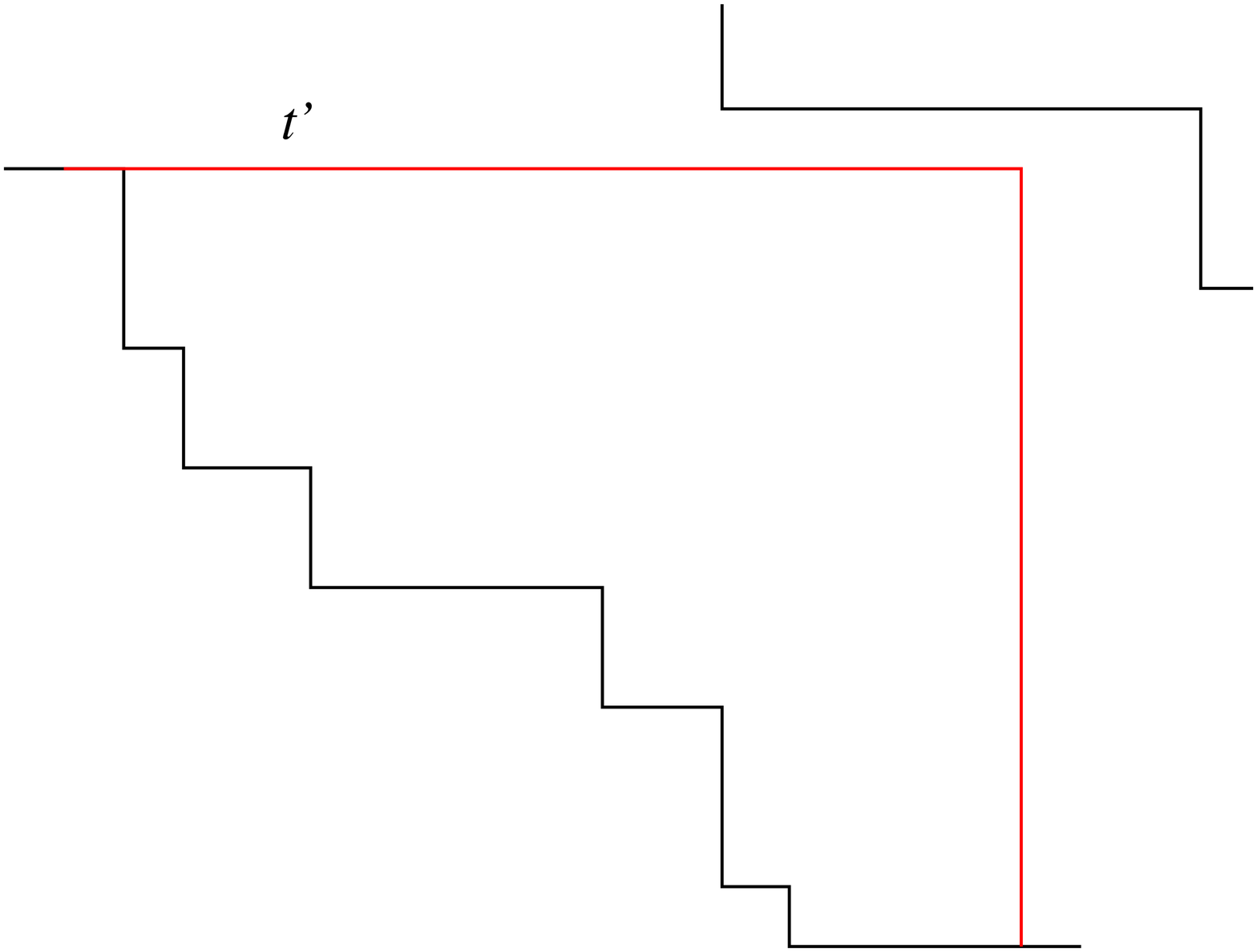} \\
  {\bf (a)}   &  & {\bf(b)}  \\
  \includegraphics[width=.45\textwidth]{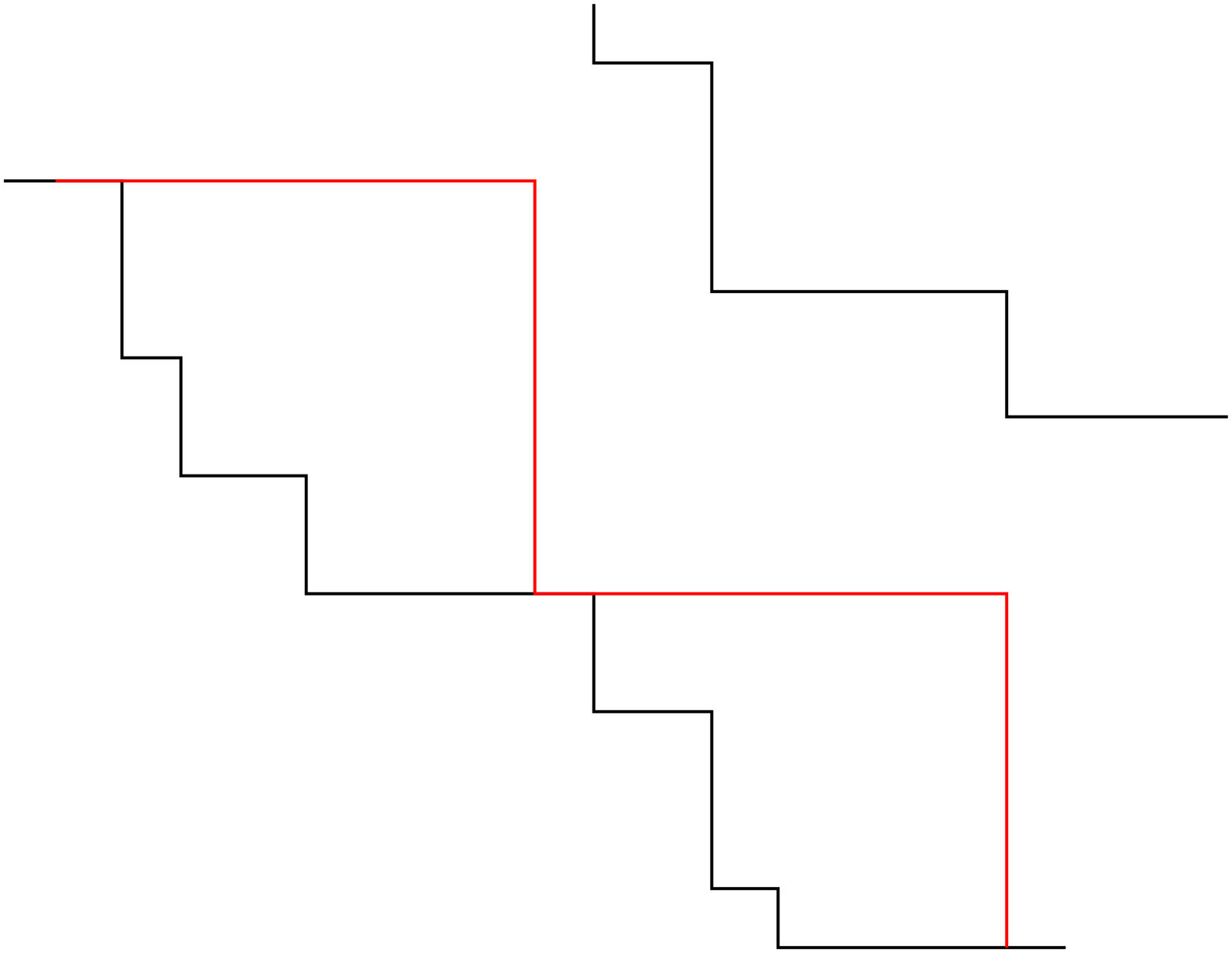} & \hspace*{.7cm} &
  \includegraphics[width=.45\textwidth]{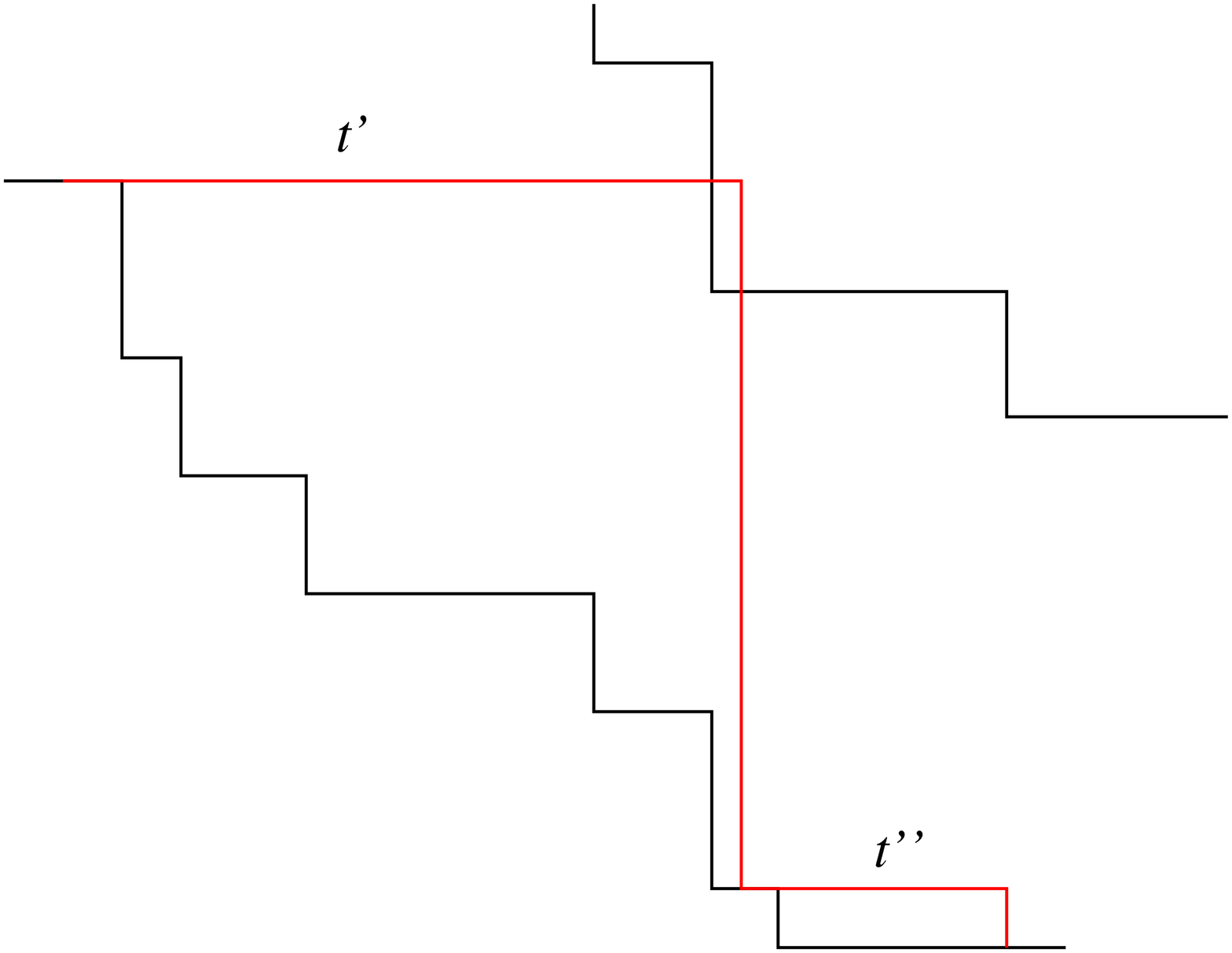} \\
  {\bf (c)}   &  & {\bf(c)}  \\
  \end{tabular}
  \caption{{Maintaining Invariant~\ref{inv:empty}. For simplicity only 
      several segments of $\cM_{id}$, $\cM_{id-g}$, and $\cB_i$ are shown; 
      $\cB_i$ is drawn in red color. \bf (a) \& (b):} Two empty segments are
    replaced with one (possibly empty) segment $t'$. 
    {\bf (c) \& (d):}  Two empty segments are replaced with one non-empty
    segment $t'$ and one empty segment $t''$. 
  }
  \label{fig:2empty}
\end{figure}
\newpage

\section*{Appendix B. Proof of Lemma~\ref{lemma:stair}}

In the first part of this section we describe the construction procedure 
of a boundary $\cB_i$. 
Then, we will prove some facts about 
$\cB_i$ and describe the update procedure. In the last part of this section 
we will prove that all $\cB_i$ are
 updated $O(1)$ times for $d$ updates of $\cM_i$. 

{\bf Construction of  Additional  Staircases.}
We construct one staircase $\cB_i$ for $d$ staircases 
$\cM_{(i-1)d+1},\ldots,\cM_{id}$. 
Let $p$ be the starting point of the
 staircase $\cM_{id}$, i.e., $p\in \cM_i$ and $p.x=1$. 
The staircase  $\cB_i$ is the path traced by $p$ as we alternatively
 move $p$ in the $+x$ and $-y$ direction until it hits the $x$-axis.

A segment $s$ \emph{covers} a point $p$ if $\sleft(x) \leq p.x \leq \sright(s)$.
A segment $r$ is \emph{related} to a segment $s$ if $s$ covers the left 
endpoint of $r$; a segment $s$ \emph{covers}  a segment $r$ if
 $\sleft(s)\leq \sleft(r)$
 and $\sright(r)\leq \sright(s)$. A point $p$ \emph{dominates} 
a segment $s$ if $p$ dominates the left endpoint of $s$. A segment $s$ 
\emph{follows} the segment $r$ in a staircase $\cB_i$ or $\cM_i$ 
(resp.\ $r$ \emph{precedes} s) if 
both $r$ and $s$ belong to the same staircase and $\sright(r)=\sleft(s)$.

Let $g=d/2$. For convenience we assume that each point $q\in S$ has 
even $x$-coordinate. This is achieved by replacing each point 
$q=(q.x,q.y)$ with a point $q'=(2q.x,q.y)$. Endpoints of all 
segments of $\cB_i$ will have odd $x$-coordinates.
The set $G_i$ contains all  segments of 
$\cM_{id-g},\ldots,\cM_{id}$. 
The staircase is constructed by repeating the following steps 
until $p$ hits the $x$-axis or the $x$-coordinate of $p$ is maximal possible,
 i.e. until $p.y=0$ or $p.x=2n$: \\
(1) We move $p$ in the $+x$ direction until $p$ cuts $\cM_{id-g}$, 
i.e until $p.x=\sleft(s)+1$ for a segment $s\in\cM_{id-g}$ such that 
$y(s)<p.y$\\
(2) If $p.x< 2n$, we move $p$ in $-y$ direction until it hits 
a segment of $\cM_{id}$ or $p.y =0$. \\
Observe that at the beginning of step $(1)$ the point $p$ always belongs 
to a horizontal segment of $\cM_{id}$. Hence, a point on $\cM_{id+1}$ does 
not dominate a segment of $\cB_i$. 
Since each horizontal segment 
of $\cB_i$ cuts $\cM_{id-g}$ it also cuts $\cM_{id-j}$, $0< j < g$. 
Hence, there are at least $g$ segments of $G_i$ related to each horizontal 
segment of $\cB_i$ and the total number of segments in all $\cB_i$ 
is $O(\frac{n}{g})$.  
An example of 
a (just constructed) additional staircase is shown on
 Fig.~\ref{fig:stair1}. \\
{\bf Updates.} 
When we update a staircase $\cM_{id+j}$ for $g/2\geq j \geq -g/2$
 by operation $\replace$, the 
staircase is moved in the north-east direction. As a result, 
a point on a staircase $\cM_{id+j}$, $j>0$, may dominate a segment 
of $\cB_i$. Therefore  we maintain a weaker property: 
no segment of $\cM_{id+g}$ 
dominates $\cB_i$ and each point of $\cB_i$ is dominated by a 
point on $\cM_{id-g}$.   Our goal is to update $\cB_i$ 
$O(1)$ times for $\Omega(g)$ updates of $\cM_j$ (in average). 
We achieve this by maintaining the following invariants
\begin{invariant}\label{inv:dom1}
Each segment $s\in \cB_i$ is dominated by  the right endpoint of a segment 
$r\in \cM_{id}$.
\end{invariant}
\begin{invariant}\label{inv:dom2}
No point of $\cB_i$ is dominated by a point of $\cM_{id+g/2}$.
\end{invariant}
\begin{invariant}\label{inv:cut}
No segment $s\in \cB_i$ cuts $\cM_{id-g+1}$.
\end{invariant}
We say that a segment $s$ is empty if it does not cut $\cM_{id-g/2}$.
\begin{invariant}\label{inv:empty}
If a segment $s_2$ follows $s_1$ in $\cB_i$, then 
either $s_2$ or $s_1$ is not empty. 
\end{invariant}
If Invariants~\ref{inv:dom1} and~\ref{inv:cut} are true when $\cB_i$ 
is constructed, they will not be violated after updates of $\cM_{id},\ldots,
\cM_{id-g}$. 
We update $\cB_i$ if Invariants~\ref{inv:dom2} or~\ref{inv:empty}
are violated: If a  segment $s\in \cB_i$, such that $s$ was not empty
when $s$ was inserted into $\cB_i$,  
does not cut $\cM_{id-g/2}$ after an operation $\replace(q,id-g/2)$, 
we call the procedure $\Rectify(\cB_i,s)$ that will be described later in this section. 
If a segment $s$ of $\cB_i$ is dominated 
by a point of $\cM_{id+g/2}$ after $\replace(q,id+g/2)$, 
we also call the procedure $\Rectify(\cB_i,s)$.
\begin{fact}\label{fact:dom1}
If a point $q$ dominates a segment of $\cM_{id-j}$, then $q$ dominates at 
least one segment of $\cM_{id-k}$ for each $k<j$.
\end{fact}
\begin{fact}\label{fact:dom2}
If a point $q$ dominates more than two  segments of $\cB_{i}$, then 
$q$ dominates at least one segment of $\cM_{id-j}$ for each 
$j=1,2,\ldots,g/2$. 
\end{fact}
\begin{proof}
If $q$ dominates three segments of $\cB_i$ then $q$ dominates the right 
endpoint of at least  
one non-empty segment $s\in \cB_i$. Since $s$ cuts $\cM_{id-g/2}$, 
the right endpoint of $s$ dominates a segment of $\cM_{id-g/2}$. 
Hence, the right endpoint of $s$ dominates  at least one segment 
of $\cM_{id-j}$ for $j=1,2,\ldots,g/2-1$ by Fact~\ref{fact:dom1}. 
Since $q$ dominates the right endpoint of $s$, $q$ also dominates 
at least one segment of $\cM_{id-j}$ for each $j=1,2,\ldots,g/2$.
\end{proof}

\begin{fact}\label{fact:dom3}
Any  point $q$ on $\cM_{id+j}$, $j\geq 0$, dominates at most two 
segments of $\cB_i$. 
\end{fact}
\begin{proof}
Suppose that a point  $q$ on $\cM_{id+j}$ dominates more than two  segments 
of $\cB_{i}$. Then, there is a point $q'$ on $\cM_{id}$ that 
also dominates more than two  segments 
of $\cB_{i}$. By Fact~\ref{fact:dom2},  $q'$ dominates a segment of 
$\cM_{id-j}$ for each $j=1,2,\ldots,g/2$. 
Since a point on $\cM_{id}$ cannot dominate 
a point on $\cM_{id-g/2}$, we obtain a contradiction.  
\end{proof}
Fact~\ref{fact:dom3}, which is a corollary of  Invariant~\ref{inv:empty},
guarantees us that each operation $\replace(q,id+j)$ such that $q$ dominates
$\cB_i$ affects at most two segments of $\cB_i$. This will be important 
in our analysis of the number of updates of $\cB_i$.
Now we are ready to describe the update procedure.

The procedure $\Rectify(\cB_i,s)$ deletes a segment  $s$ and a number 
of preceding and following segments and replaces them 
with new segments. We say that a segment $s'$ is the child 
of $s$ if $s$ was removed by an operation $\replace(q)$, such 
that $q$ is the right endpoint of $s'$; $s'$ is a \emph{descendant} of $s$ 
if $s'$ is a child of $s$ or a descendant of a child of $s$.
Let $s_0$ be the segment that precedes $s$ in $\cB_i$. 
Let $s_1,s_2,\ldots$ be segments of $\cB_i$ such that $s_1$ follows 
$s$ and $s_i$ follows $s_{i-1}$ for $i>1$.
Suppose that $s$ contained the right endpoint of a segment $r_B\in \cM_{id}$
that belonged to $\cM_{id}$ when $s$ was inserted into $\cB_i$, 
and let $r_u$ be the descendant of $r_B$ that belongs to $\cM_{id}$ 
when the procedure $\Rectify(\cB_i,s)$ is performed. 
By Fact~\ref{fact:dom3}, $r_u$ may dominate the segment $s_0$ that 
precedes $s$ in $\cB_i$, but $r_u$ does not dominate the segment 
that precedes $s_0$ in $\cB_i$. 
Let $r_{\max}$ be a descendant of a segment $r\in \cM_{id}$ related 
to $s$ with the largest $x$-coordinate of its right endpoint. 
By Fact~\ref{fact:dom3}, $r_{\max}$ may dominate $s_1$ and $s_2$
but it cannot dominate $s_3$.

Now we must decide which segments are to be deleted from $\cB_i$ and 
how to construct new segments. We delete segments that are dominated 
by $r_u$ or $r_{\max}$. As shown above, there are at most three 
such segments (except of $s$ itself). If $s_f$, $f\leq 2$, is the last 
segment dominated by $r_{\max}$, we may also remove some segments that
follow $s_f$. But our guarantee is that 
all removed segments $s_{f+1},\ldots s_m$ do not cut $\cM_{id-3g/4}$.
We insert new segments into $\cB_i$ by moving a point $p$ 
in $+x$ and $-y$ directions.  A more detailed description 
follows.

To simplify the description, we will use  set  $\cV_i$ 
that  contains some horizontal 
segments that currently  belong to  $\cM_{id}$ and some segments 
that belonged to $\cM_{id}$ but are already deleted. 
When a staircase $\cB_i$ is constructed, $\cV_i$ contains all 
horizontal segments of $\cM_{id}$. When the procedure $\Rectify(\cB_i,s)$ 
is called, we delete all segments of $\cV_{i}$ dominated by $r_u$ or 
$r_{\max}$ and insert all segments $r\in \cM_{id}$, such that 
$\sright(r_u) \leq  \sleft(r) \leq \sleft(r_{\max})$. Segments 
of $\cV_i$ are used to  ``bound the staircase $\cB_i$ from below'', 
i.e. the left endpoint of each horizontal segment in $\cB_i$ belongs 
to a segment from $\cV_i$.  \\
(1)  
Let $p$ be the point on $\cB_i$ such that $p.y = y(r_u)$. This is the left 
endpoint of the first inserted segment of $\Rectify(\cB_i,s)$.\\
(2) We move $p$ in the $+x$ direction until $p.x=\sright(r_{\max})$ or $p$ 
cuts $\cM_{id-g}$. While  $p.x < \sright(r_{\max})$, we repeat the following 
steps: we move $p$ in the $-y$ direction until it hits ``new'' 
$\cM_{id}$; then, we move $p$ in $+x$ direction until it 
cuts $\cM_{id-g}$. Observe that all horizontal segments inserted in step 2 
cut $\cM_{id-g}$.\\
(3) When  $p.x=\sright(r_{\max})$, we move $p$ in $+x$ direction until 
it cuts $\cM_{id-3g/4}$ and $p.x\geq \sright(r_{\max})+1$. 
Suppose 
that now $\sleft(s_m)< p.x \leq \sright(s_m)$ for some $s_m$ in $\cB_i$. 
We continue 
moving $p$ in $+x$ direction until $p.x=\sright(s_m)$ or $p$ cuts 
$\cM_{id-g}$. Then, we move $p$ in $-y$ direction until $p$ hits 
$\cV_i$\\ 
(4) We insert a new segment $t$ instead of the segment $s_m$. It is possible 
that now $p.y < y(s_m)$.
We move $p$ in $+x$ direction until  $p.x=\sright(s_{m})$ and move $p$ in 
$-y$ direction until $p$ hits $\cV_i$. \\  
(5) Now we must pay attention that  Invariant~\ref{inv:empty} is
 maintained: all inserted segments except of may be the last one are not 
empty. 
Let $t$ denote the last inserted segment, and suppose that both $t$ 
and $s_{m+1}$ are empty. We can replace two empty segments 
either with one empty segment, or one non-empty and one empty segment
as follows. We replace  $t$ with a new segment $t'$: 
the left endpoint of $t'$ coincides with the left endpoint of $t$ and 
either  $\sright(t')=\sright(s_{m+1})$ or  
$t'$ cuts $\cM_{id-g}$. 
If $t'$ cuts $\cM_{id-g}$, we replace $s_{m+1}$ with a new 
segment $t''$ such that $\sleft(t'')=\sright(t')$, $\sright(t'')=\sright(s_{m+1})$, 
and $y(t'')=y(s_{m+1})$. See Fig.~\ref{fig:2empty} for an example of
 step (5).\\
Observe that all but one non-empty segments constructed by 
$\Rectify(\cB_i,s)$ cut $\cM_{id-g}$. The only exception is the segment 
constructed in step (3) that cuts $\cM_{id-3g/4}$. 
We prove in Appendix C 
that the total number of updates of $\cB_i$ 
during the execution of the sweep plane algorithm is $O(\frac{n}{d})$. 

This completes the Proof  of Lemma~\ref{lemma:stair}.

\section*{Appendix C. Analysis of Update Operations for Additional Staircases}
\label{sec:analys}
We will show below  that the data structure $\cD$ is updated $O(\frac{n}{g})$ 
times during the execution of the sweep-plane algorithm.
First, we will estimate the number of deleted segments.
We will estimate the number of insertions in the end of this section.
We assign $c$ credit points to each segment of $\cM_{id-j}$ and 
$3c$ credit points to every segment of $\cM_{id+j}$ for
$j=1,2,\ldots,g$ and $c= 24$.
Insertion of a new segment into $\cB_i$ is free and deletion 
costs $g$ credit points. 

Every time when we perform operation $\replace(q,id+j)$ for 
$j>0$ we distribute the credit points of 
the newly inserted segment $r$ with right endpoint $q$ among several segments
of $\cB_i$. We evenly distribute credits of $r$ among segments 
$s_j\in B$ such that either $q$ dominates $s_j$ or $q.x> \sleft(s_j)$ 
and $q.y \geq y(s_{j+1})$ where $s_{j+1}$ denotes the segment 
that follows $s_j$ in $\cB_i$. By Fact~\ref{fact:dom2} there are at most 
three such segments $s_j$; hence each $s_j$ obtains at least $c$ credits. 
When a segment $r$ of $\cM_{id-j}$, $j\geq 0$, is deleted, 
we assign credits of $r$ to $s\in \cB_i$, such that $r$ is related 
to $s$. 
We say that a segment $s$ is initially non-empty if $s$ cuts 
$\cM_{id-g/2}$ when $s$ is inserted into $\cB_i$. We will 
show below that we can pay $g$ credit points for each deleted segment 
of $\cB_i$ and maintain the following property.
\begin{property}\label{prop:cred}
Every initially non-empty segment $s\in \cB_i$ 
that does not cut $\cM_{id-g+k}$, $k\geq g/4$, 
accumulated at least $c\cdot k$ credit points.\\ 
Every segment $s\in \cB_i$ that is dominated by a point on $\cM_{id+j}$ 
accumulated $j\cdot c$ credit points.
\end{property}
\begin{proof}
Property~\ref{prop:cred} is obviously true for a just constructed 
staircase $\cB_i$.
Suppose that Property~\ref{prop:cred} is true after the procedure 
$\Rectify(\cB_i,s)$ was called for segments of $\cB_i$ $f\geq 0$ times. 
We will 
show that  this property is maintained after the $(f+1)$-th call of the 
procedure $\Rectify(\cB_i,s)$. 
If an initially non-empty segment $s$ does not cut $\cM_{id-g+k'}$ after 
the $f$-th call of $\Rectify$ is completed, then $s$ accumulated $k'c$ credits. 
If the segment $s$ does not cut $\cM_{id-g+k}$ at some point after the 
$f$-th call of $\Rectify$, then at least $k-k'$ segments of $\cM_{id-g+j}$, 
$k' < j\leq k$, that are related to 
 $s$ are already deleted. 
Hence, $s$ accumulated at least $ck'+ c(k-k')=ck$ credits.
Therefore if an initially non-empty  segment $s$ does not cut $\cM_{id-g/2}$, 
then $s$ accumulated $cg/2$ credits. 

If a point of $\cM_{id+k'}$ dominates a segment $s\in \cB_i$ when 
the $f$-th call of the procedure $\Rectify$ is completed, then 
$s$ has $k'\cdot c$ credit points. 
If a point of $\cM_{id+k}$ dominates a segment $s\in \cB_i$, then 
we performed at least one operation $\replace(q,id+j)$ such that 
$q$ dominates $s$ for each $k' <  j \leq k$. 
Hence, $s$ accumulated $kc$ credits. 
Therefore if a segment of $\cB_i$ is dominated by a point on 
$\cM_{id+g/2}$, then $s$ has $cg/2$ credits.

Hence, when we start the procedure $\Rectify(\cB_i,s)$, the segment $s$ 
has $cg/2$ credit points. 
In addition to $s$, we may have to remove segments $s_0,s_1,s_2$ because 
a descendant of some segment $r\in \cM_{id}$, such that $r$ was related 
to $s$ when $s$ was constructed, dominates $s_0$, $s_1$, or $s_2$. 
If segments $s_3,s_4,\ldots s_m$ are removed by $\Rectify(\cB_i,s)$, then 
each non-empty segment among $s_3,\ldots,s_m$ does not cut 
$\cM_{id-3g/4}$. By Property~\ref{prop:cred}, 
every such segment has $cg/4$ credits. 
Since there are at least $(m-4)/2$ non-empty segments among 
$s_3,\ldots, s_m$, we can use $cg(m-4)/4$ 
credits accumulated by non-empty segments 
to remove $s_4,\ldots, s_m$. We use $cg/4=6g$ credits 
 accumulated by $s$ to remove $s$ and to remove 
$s_0$,$s_1$,$s_2$,$s_3$, and $s_{m+1}$ if necessary. 
If a segment $s'$ inserted after the  procedure $\Rectify$  cuts 
only $3g/4$ staircases $\cM_{id+j}$, then we transfer to $s'$ 
 the remaining credit points accumulated by $s$. 
Recall that there is at most one such segment $s'$ that may be inserted during 
step (3) of the update procedure. 
Since $s$ accumulated
at least $cg/2$ credit points and at most $cg/4$ are spent for 
removing the segments $s$, $s_0$, $s_1$, and $s_2$, 
the segment $s'$ obtains at least $cg/4$ credit points 
after the update procedure.

We must also take care of the segment $t$ resp. segments $t'$ and $t''$.
Since it is possible that $y(t)< y(s_m)$, $t$ can be dominated by a point
 of $\cM_{id+j}$ 
for some $j>0$ when it is constructed.
Remaining credit points of segment $s_m$ are transferred 
to $t$ (resp.\ to $t'$);
if $t''$ is constructed, then credit points of $s_{m+1}$ are transferred to 
$t''$. If $t'$ is constructed but $t''$ is not constructed (i.e. 
if $t'$ replaces both $s_m$ and $s_{m+1}$), then 
credits of $s_{m+1}$ are also transferred to $t'$.  
If $t$ is dominated by a point of $\cM_{id+k}$ for some $0<k< g/2$, 
then we performed an operation $\replace(q, id+j)$ for each $j\leq k$, 
such that $q.y \geq y(t)$ and $q.x \geq \sleft(t)$. 
Since $\sleft(t)> \sleft(s_m)$ and $y(t)> y(s_{m+1})$, 
$s_m$ was assigned $c$ credits for each $1\leq j\leq k$. 
Hence, if $t$ is dominated by $\cM_{id+j}$, then $t$ has 
at least $cj$ credit points.   
The same is also true for $t'$ and $t''$.
\end{proof}

We can conclude from Property~\ref{prop:cred}  
that we can always pay $g$ 
credit points for a deleted segment of $\cB_i$; hence, the 
total number of deleted segments is $O(\frac{n}{g})$.


Let $N_f$ be the number of segments in all $\cB_i$ when the algorithm 
is finished. Since for every second segment $s$ in $\cB_i$ there 
are at least $g/2$ segments of $\cM_{id},\ldots,\cM_{id-g}$ related 
to $s$, the total number of segments in $\cB_i$ is $O(n_i/g)$ 
where $n_i$ denotes the total number of segments in 
$\cM_{id},\ldots,\cM_{id-g}$. 
Hence, $N_f=O(\sum n_i/g)=O(n/g)$. 
Clearly $N_i=N_f+N_d$ where $N_i$ is the number of inserted segments 
and $N_d$ is the number of deleted segments.
Hence, $N_i=O(n/g)$ and the total number of inserted and deleted segments 
in all $\cB_i$ is $N_i+N_d=O(n/g)=O(n/d)$.

\section*{Appendix D.  Staircases on $U\times U$ Grid}
\begin{lemma}\label{lemma:slow2}
We can store $n$ horizontal staircase segments with endpoints 
on $U\times U$ grid in a $O(n\log n)$ space data
 structure that answers ray shooting queries in 
$O(\log \log U + (\log \log n)^3)$ time 
and supports operation $\replace(q,i)$  in $O(m\log^3 n(\log\log n)^2)$ time 
where $m$ is the number of segments inserted into and deleted from the 
staircase $\cM_i$, and operation $\new(q)$ in $O(\log^3 n(\log\log n)^2)$ 
time.  
\end{lemma}
\begin{proof}
Instead of storing point coordinates of segment endpoints 
in the data structure, we store 
\emph{labels} of point coordinates: each $x-$ and $y$-coordinate is assigned 
an $x$-label ($y$-label), so that the $x$-label ($y$-label) 
of $q$ is smaller than 
the $x$-label ($y$-label) of $p$ if and only if $q.x < p.x$ ($q.y < p.y$). 
All labels belong to range $[1,O(n)]$ and  are maintained 
using the technique of~\cite{IKR81,willard92}. 
When a new segment is inserted or deleted, $O(\log^2 n)$ labels may
 change, and  we have to delete and re-insert into data 
structures those segments whose labels are 
changed. Since each segment is stored in $O(\log n)$ secondary data 
structures, a deleted/inserted segment leads to $O(\log^3 n)$ updates 
in $\cL_v$ and $\cR_v$. 
Hence the update time is $O(\log^3 n(\log\log n)^2)$. 
The query procedure is exactly the same as in the proof of
 Lemma~\ref{lemma:slow}.
\end{proof}

\end{document}